\documentclass[usenames,dvipsnames,runningheads]{llncs}
\usepackage[utf8]{inputenc}
\usepackage{graphicx}

\usepackage{amsthm,amssymb,amsmath}
\usepackage{gensymb}
\usepackage{tikz}
\usetikzlibrary{decorations.pathmorphing}
\usetikzlibrary{decorations.markings}
\tikzset{city/.style={black, circle, draw, inner sep=1.2pt}}
\usepackage{lineno}
\usepackage{placeins}
\usepackage{comment}
\usepackage{wrapfig}
\usepackage{subcaption}
\usepackage{rotating}
\usepackage{soul}
\usepackage[margin=4cm]{geometry}
\usepackage{varioref}
\usepackage{algorithm}% http://ctan.org/pkg/algorithms
\usepackage{float}
\usepackage[toc,page]{appendix}
\usepackage{appendix}
\usepackage{changepage}
\usepackage[framemethod=TikZ]{mdframed}
\usepackage{lipsum}
\usepackage[noend]{algpseudocode}
\usepackage{xspace}
\usepackage[capitalise]{cleveref}
\usepackage[color=green!10!white,textsize=small]{todonotes}

\newcommand{\TSP}{{\sc TSP}\xspace}
\newcommand{\I}{\ensuremath{\cal I}\xspace}

\newcommand{\MetricTSP}{{\sc Metric TSP}\xspace}

\newcommand{\DubinsTSP}{{\sc Dubins TSP}\xspace}
\newcommand{\ExactCover}{{\sc Exact Cover}\xspace}
\newcommand{\AStar}{\ensuremath{A^*}\xspace}
\newcommand{\multipointastar}{\texttt{Multi Point A$^*$}\xspace}
\newcommand{\flipTour}{\texttt{FlipTour}\xspace}
\newcommand{\C}{\ensuremath{{\cal C}}\xspace}
\newcommand{\successors}{\texttt{succ}\xspace}
\newcommand{\VectorTSP}{{\sc Vector TSP}\xspace}
\newcommand{\VectorTSPWS}{{\sc Vector TSP With Stops}\xspace}
\newcommand{\EuclideanTSP}{{\sc Euclidean TSP}\xspace}
\newcommand{\Racetrack}{{\sc Racetrack}\xspace}

\newcommand{\ContinuousVectorTSPWithStops}{{\sc Continuous VectorTSP with Stops}\xspace}
\newcommand{\ConstructionOne}{Construction 1\xspace}
\newcommand{\ConstructionTwo}{Construction 2\xspace}
\newcommand{\ConstructionThree}{Construction 3\xspace}

\theoremstyle{definition}

\numberwithin{subcase}{case}
\usepackage[numbers]{natbib}
\usetikzlibrary{calc}

\begin{document}

\title{\textsc{Vector TSP}: A Traveling Salesperson Problem with Racetrack-like acceleration constraints \thanks{A preliminary version of this work was presented at ALGOSENSORS 2020. This work has been supported by ANR project TEMPOGRAL (ANR-22-CE48-0001).}}

\author{Arnaud Casteigts\inst{1,2}
  \and
	Mathieu Raffinot\inst{2}  \and 
	Mikhail Raskin\inst{2}  \and 
	Jason Schoeters\inst{3}
}

\institute{
  University of Geneva, Switzerland
  \and
  LaBRI, Université de Bordeaux, CNRS, Bordeaux INP, France
  \and
  University of Cambridge, United Kingdom\medskip\\
\email{\{arnaud.casteigts, mathieu.raffinot, mikhail.raskin\}@labri.fr}\\
\email{js2807@cam.ac.uk}
}

\date{}
\maketitle

{\centering\footnotesize \textit{We dedicate this paper to Ralf Klasing for his 60th birthday. Ralf had many contributions to the \TSP, addressing both approximation and exact algorithms, and he also considered less common versions of the problem, such as the \textit{Bamboo Garden Trimming Problem}, which certainly shares the feature of being as exotic as Vector TSP. Ralf, we hope you will enjoy, happy birthday!}\par}

\begin{abstract}
We study a new version of the Traveling Salesperson Problem, called \VectorTSP, where the traveler is subject to discrete acceleration constraints, as defined in the paper-and-pencil game Racetrack (also known as Vector Racer). In this model, the degrees of freedom at a certain point in time depends on the current velocity, and the speed is not limited. 
  
The paper introduces this problem and initiates its study, discussing also the main differences with existing versions of TSP. Not surprisingly, the problem turns out to be NP-hard. A key feature of \VectorTSP is that it deals with acceleration in a discrete, combinatorial way, making the problem more amenable to algorithmic investigation. The problem involves two layers of trajectory planning: (1) the order in which cities are visited, and (2) the physical trajectory realizing such a visit, both interacting with each other. This interaction is formalized as an interactive protocol between a high-level tour algorithm and a trajectory oracle, the former calling the latter repeatedly. We present an exact implementation of the trajectory oracle, adapting the A* algorithm for paths over multiple checkpoints whose ordering is \emph{given} (this algorithm being possibly of independent interest). To motivate the problem further, we perform experiments showing that the naive approach consisting of solving the instance as an \EuclideanTSP first, then optimizing the trajectory of the resulting tour, is typically suboptimal and outperformed by simple (but dedicated) heuristics.

\end{abstract}

\section{Introduction}
\label{sec:intro}
The problem of visiting a given set of locations and returning to the starting point, while minimizing the total cost, is known as the \textsc{Traveling Salesperson Problem} (\TSP, for short). The problem was independently formulated by Hamilton and Kirkman in the 1800s~\cite{biggs1986graph} and has been extensively studied since then.
Many versions of this problem exist, motivated by applications in fields as varied as delivery planning~\cite{chalasani1999approximating}, stock cutting~\cite{sonawane2020optimizing}, and DNA computation/reconstruction~\cite{mallen2013dna, narayanan1998dna}. 
In the original version, an instance of the problem is equivalent to a graph whose vertices represent the {\em cities} (places to be visited) and whose weights on the edges represent the cost of moving from one city to another. The goal is to find the minimum cost tour (optimization version) or to decide whether a tour having at most some cost exists (decision version) subject to the constraint that every city is visited {\em exactly} once.
Karp proved that the Hamiltonian Cycle problem is NP-hard, which implies that TSP is NP-hard~\cite{karp1972reducibility}.
On the positive side, while the trivial algorithm has a factorial running time (considering all the possible tours), Held and Karp presented in~\cite{heldkarp62} a dynamic programming algorithm running in time $O(n^22^n)$, which as of today remains the fastest known deterministic algorithm for TSP (a faster randomized algorithm exists~\cite{bjorklund}).
TSP was shown to be inapproximable to within a constant factor by Orponen and Manilla in 1990~\cite{orponenmannila90}. 
Some generalizations of TSP include a version with unknown environment which is discovered while visiting the cities, known as the \textsc{Covering Canadian Traveler Problem}~\cite{canadian1,canadian2}, and more recently, a version operating on a temporal graph where costs between cities are time-dependent, known as \textsc{Time Dependent TSP}~\cite{fontaine:hal-03865036}. Another problem closely related to the \TSP is the \textsc{Bamboo Garden Trimming Problem} \cite{gkasieniec2017bamboo,kuszmaul2022bamboo}, where one aims to maintain the height of a forest of bamboos to a minimum.

In many cases, the problem can be restricted to a more tractable setting.
In \MetricTSP, the costs must respect the triangle inequality, namely $cost(u,v) \le cost(u,w) + cost(w,v)$ for all $u,v,w$, or equivalently, the constraint of visiting a city exactly once is relaxed.
\MetricTSP was shown to be approximable within a ratio of $1.5$ by Christofides~\cite{christofides76}. 
This ratio has been slightly improved recently by Karlin \textit{et al.} in~\cite{karlin2021slightly}. A lower bound of $1.0045$ is known for the approximation ratio (i.e., no PTAS exists for \MetricTSP~\cite{PV06}).
A particular case of \MetricTSP is the \EuclideanTSP, where cities are points in the plane, and the weights are the Euclidean distance between them. This problem is still NP-hard (see Papadimitriou~\cite{papadimitriou1977euclidean} and Garey \textit{et al.}~\cite{garey1976some}), but admits a PTAS, discovered independently by Arora~\cite{arora96} and Mitchell~\cite{mitchell99}.

One attempt to add physical constraints to the \EuclideanTSP is \DubinsTSP~\citep{le2007curvature}, which encode inertia through bounding the curvature of a trajectory by some radius, thereby offering a geometric abstraction to the problem. However, the radius being a fixed value, the speed is considered constant.
Savla \textit{et al.} present multiple results on \DubinsTSP, such as an elegant algorithm in~\cite{savla2005point} which modifies every other segment of an \EuclideanTSP solution to enforce the curvature constraint. 

More flexible models for acceleration have been considered in the context of path planning problems, where one aims to find an optimal trajectory between \emph{two} given locations (typically, with obstacles), while satisfying acceleration constraints. More generally, the literature on {\it kinodynamics} is pretty vast (see, e.g.~\cite{canny1988complexity,canny1991exact,donald1993kinodynamic}). The constraints are often formulated in terms of a number of dimensions, a bounded acceleration and a bounded speed. The positions are typically represented by real coordinates, using tools from control theory and analytic functions~\cite{https://doi.org/10.48550/arxiv.2005.03186, savla2009traveling, yakowitz1984computational}.  

In a recreational column of the {\em Scientific American} in 1973~\cite{gardner1973sim}, Martin Gardner presented a paper-and-pencil game known as \Racetrack (not to be confused with a similarly titled \TSP heuristic~\cite{mobile_sink} that is not related to acceleration). The physical model is as follows. In each step, a vehicle moves according to a discrete-coordinate vector (initially the zero vector), with the constraint that the vector at step $i+1$ cannot differ from the vector at step $i$ by more than one unit in each dimension. The game consists of finding the best trajectory (smallest number of vectors, regardless of distance) in a given race track. A nice feature of such models is the ability to think of the state of the vehicle at a given time as a point $(x,y,dx,dy)$ in a {\em configuration space}.
Note that this game is not a version of TSP, since only the starting and final positions are specified. In fact, an optimal trajectory in this game can be found by performing a simple breadth-first search in the configuration graph. These techniques, which we describe in more details in the present paper, were rediscovered many times, both in the \Racetrack context (see {\it e.g.}~\cite{bekos2018algorithms, blogernie, olsson2011genetic, schmid2005vector}) and in the kinodynamics literature (see {\it e.g.}~\cite{canny1988complexity,donald1993kinodynamic})---they can now be considered as folklore.
Note that a more efficient algorithm for \Racetrack instances of uniform width was proposed by Bekos \textit{et al.}~\cite{bekos2018algorithms}, as well as ''local view'' algorithms in which the track is discovered in an online fashion during the race.
The theoretical complexity of the \Racetrack problem has also been studied. In particular, Holzer and McKenzie~\cite{holzer2010computational} proved that \Racetrack is NL-complete. They also proved that the reachability problem with a given time limit (referred to as single-player \Racetrack) is NL-complete, and that deciding the existence of a winning strategy in Gardner’s original game when the positions cannot be occupied by two players simultaneously is P-complete (under polylogarithmic time reductions).

\subsection{Contributions}

In this article, we combine TSP and \Racetrack into a single problem called \VectorTSP, where a vehicle must visit an (initially unordered) set of coordinates (cities), using a trajectory that obeys racetrack-like constraints. The cost measure is the total number of vectors used, regardless of the distance traveled.
Motivations for such a problems a varied. On the practical side, there is an increasing interest for using UAVs in delivery planning and patrolling. In the longer term, outer space travel will also be subject to similar isotropic constraints, where inertia matters more than speed. On the theoretical side, this problem is quite natural and the constraints are simplest possible in a discrete setting. Yet, it requires solutions that are significantly different from its standard versions, making this problem intriguing.

The paper starts by describing the \Racetrack physical model and some of its basic properties, as well as the concept of configuration graph and algorithmic tricks based on this graph. Then, we define the \VectorTSP problem, together with a number of preliminary observations. In particular, we investigate the differences between \VectorTSP and \EuclideanTSP, showing that solutions visiting the cities according to an optimal \EuclideanTSP tour are suboptimal for \VectorTSP. This motivates the study of \VectorTSP as a problem of its own.

Then, we prove that \VectorTSP is NP-hard. This result is not surprising in itself, but the proof is far from trivial. It proceeds by defining another version of the problem called \VectorTSPWS, where every city must be visited at zero speed. This problem is closer in spirit to \EuclideanTSP, making it possible to adapt the arguments used by Papadimitriou's reduction from \ExactCover to \EuclideanTSP. Then, we show that \VectorTSPWS reduces to \VectorTSP. 

On the positive side, we present an algorithmic framework for tackling \VectorTSP, based on an interactive scheme between a high-level algorithm and a trajectory oracle. The first is responsible for exploring the space of possible visit orders, while making queries to the second for assessing the quality of a particular visit order. We present algorithms for both. The high-level algorithm adapts known heuristics for \EuclideanTSP, trying to gradually improve the solution through performing 2-permutations in the tour until a local optimum is found. As for the trajectory oracle, we present an original algorithm that generalizes \AStar to the search of multi-point paths in the configuration space, with a cost estimation function that uses unidimensional projections of the distances. 

We present experimental results based on this framework, giving empirical evidence that our algorithmic framework, as simple as it is, already outperforms solutions based on an optimal resolution of \EuclideanTSP. 
The goal of these experiments is not so much to evaluate the algorithms themselves as to confirm that \VectorTSP is an original problem that deserves further independent study.

\subsection{Organization of the paper}
In \cref{sec:model}, we present the \Racetrack physical model and the \VectorTSP problem, together with a number of algorithmic tricks and observations.
In \cref{sec:basic-results}, we characterize a number of basic properties of \VectorTSP and compare it with \EuclideanTSP. We prove that \VectorTSP is NP-hard in \cref{sec:complexity}. The algorithmic framework is presented in Section~\ref{sec:algorithms}, and experiments are reported in~\cref{sec:experiments}.
Finally, we conclude with some remarks and future research directions in~\cref{sec:conclusion}.

\section{Model and definitions}
\label{sec:model}

In this section, we present the \Racetrack model and some of its algorithmic features. Then, we define \VectorTSP and establish a few facts about it.

\subsection{The Racetrack model}

The \Racetrack model~\cite{gardner1973sim} is an acceleration model
where both time and space are discrete. The model specifies what
movements a vehicle can make based on its current position and
velocity. 
In the original formulation of~\cite{gardner1973sim}, the
space is two-dimensional, though the model generalizes naturally to higher dimensions.

\subsubsection*{Configuration.}

The state of the vehicle at a certain time is called a \emph{configuration}. A configuration $c$ is encoded by a position $pos(c)$
and a velocity $vel(c)$, both in $\mathbb{Z}^2$, being
possibly abbreviated as a $4$-tuple $(x, y, dx, dy)$, where $dx$ and $dy$ correspond to the movements the vehicule just made in the last move (initially $0$). Given a
configuration $c$, the constraints consist of restricting the
set of legal successors $\successors(c)$,
which are the configurations $c'$ whose velocity differ from
$vel(c)$ by at most one unit in each dimensions, and whose positions
correspond to moving the vehicle according to this chosen velocity.
(Morally, a configuration encodes a vector/movement that \emph{has just been} performed.)
Precisely, the successors of a configuration
$c_i=(x_i, y_i, dx_i, dy_i)$ are all the configurations
$c_j=(x_j, y_j, dx_j, dy_j)$ such that $|dx_j -dx_i| \le 1$,
$|dy_j -dy_i| \le 1$, $x_j=x_i+dx_j$ and $y_j=y_i+dy_j$.
\begin{figure}[h]
  \centering
  	\begin{tikzpicture}[scale=.34]
        \draw[step=1cm,lightgray!50,ultra thin] (-1.2,-.2) grid (19.2,14.3);
        \tikzstyle{every node}=[draw, circle, inner sep=.7pt]
        \path (0,1) node (x0) {};
        \path (5,1) node[draw=none, rotate=-90, inner sep=-5pt] (x1) {\includegraphics[width=.6cm]{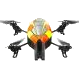}};
        \path (10,1) node[color=red] (x2bis) {};
        \path (9,2) node (x2) {};
        \path (13,3) node[color=red] (x3bis) {};
        \path (12,4) node (x3) {};
        \path (15,6) node[color=red] (x4bis) {};
        \path (14,7) node (x4) {};
        \path (16,10) node[color=red] (x5bis) {};
        \path (15,10) node (x5) {};
        \path (16,13) node[color=red!30] (x6bis) {};
        \path (15,12) node (x6) {};
        \path (15,14) node[color=white] (x7bis) {};
        \path (14,13) node (x7) {};
        \path (13,14) node[color=white] (x8bis) {};
        \path (12,14) node (x8) {};
        \draw (x0) -- (x1);
        \draw[dotted,->] (x1) -- (x2bis);
        \draw (x1) -- (x2);
        \draw[dotted,->] (x2) -- (x3bis);
        \draw (x2) -- (x3);
        \draw[dotted,->] (x3) -- (x4bis);
        \draw (x3) -- (x4);
        \draw[dotted,->] (x4) -- (x5bis);
        \draw (x4) -- (x5);
        \draw[dotted,->] (x5) -- (x6bis);
        \draw (x5) -- (x6);
        \draw[dotted,->] (x6) -- (x7bis);
        \draw (x6) -- (x7);
        \draw[dotted,->] (x7) -- (x8bis);
        \draw (x7) -- (x8);
        \tikzstyle{every node}=[draw, circle, inner sep=.7pt, color=red]
        \path (x2bis) + (1,0) node (x2bisr) {};
        \path (x2bis) + (-1,0) node (x2bisl) {};
        \path (x2bis) + (0,1) node (x2bisu) {};
        \path (x2bis) + (0,-1) node (x2bisd) {};
        \path (x3bis) + (1,0) node (x3bisr) {};
        \path (x3bis) + (-1,0) node (x3bisl) {};
        \path (x3bis) + (0,1) node (x3bisu) {};
        \path (x3bis) + (0,-1) node (x3bisd) {};
        \path (x4bis) + (1,0) node (x4bisr) {};
        \path (x4bis) + (-1,0) node (x4bisl) {};
        \path (x4bis) + (0,1) node (x4bisu) {};
        \path (x4bis) + (0,-1) node (x4bisd) {};
        \path (x5bis) + (1,0) node (x5bisr) {};
        \path (x5bis) + (-1,0) node[color=green] (x5bisl) {};
        \path (x5bis) + (0,1) node (x5bisu) {};
        \path (x5bis) + (0,-1) node (x5bisd) {};
        \path (x6bis) + (1,0) node[color=red!30] (x6bisr) {};
        \path (x6bis) + (-1,0) node[color=red!30] (x6bisl) {};
        \path (x6bis) + (0,1) node[color=red!30] (x6bisu) {};
        \path (x6bis) + (0,-1) node[color=red!30] (x6bisd) {};
        \path (x8bis) + (-1,0) node[color=green] (x8bisl) {};
        \path (x2bis) + (1,1) node (x2bisr) {};
        \path (x2bis) + (-1,1) node[color=green] (x2bisl) {};
        \path (x2bis) + (-1,-1) node (x2bisu) {};
        \path (x2bis) + (1,-1) node (x2bisd) {};
        \path (x3bis) + (1,1) node (x3bisr) {};
        \path (x3bis) + (-1,1) node[color=green] (x3bisl) {};
        \path (x3bis) + (-1,-1) node (x3bisu) {};
        \path (x3bis) + (1,-1) node (x3bisd) {};
        \path (x4bis) + (1,1) node (x4bisr) {};
        \path (x4bis) + (-1,1) node[color=green] (x4bisl) {};
        \path (x4bis) + (-1,-1) node (x4bisu) {};
        \path (x4bis) + (1,-1) node (x4bisd) {};
        \path (x5bis) + (1,1) node (x5bisr) {};
        \path (x5bis) + (-1,1) node (x5bisl) {};
        \path (x5bis) + (-1,-1) node (x5bisu) {};
        \path (x5bis) + (1,-1) node (x5bisd) {};
        \path (x6bis) + (1,1) node[color=red!30] (x6bisr) {};
        \path (x6bis) + (-1,1) node[color=red!30] (x6bisl) {};
        \path (x6bis) + (-1,-1) node[color=green] (x6bisu) {};
        \path (x6bis) + (1,-1) node[color=red!30] (x6bisd) {};
        \path (x7bis) + (-1,-1) node[color=green] (x7bisu) {};
        \path (x7bis) + (-1,-1) node[color=green] (x7bisu) {};        
      \end{tikzpicture}
      \hfill
    \caption{\label{fig:trajectories} Example of a trajectory in the \Racetrack model.}
  \end{figure}
These constraints are illustrated in Figure~\ref{fig:trajectories}.
Some versions of the model consider stronger restrictions, by allowing a single dimension to change in one time step. We refer to both versions as the $9$-successor model and the $5$-successor model, respectively. 
These models generalize naturally to $d$-dimensional spaces by applying similar constraints for each dimension. For any fixed number of dimensions $d$, a key aspect is that the number of successors is \emph{constant}. 

\subsubsection*{Trajectory.}
A trajectory is a sequence of configurations $c_1,c_2,...,c_k$, with $c_{i+1} \in \successors(c_i)$ for all $i<k$. 
The cost of a trajectory is its length $k$.
The fact that acceleration constraints are isotropic in the \Racetrack model imply that trajectories are \emph{reversible}.
The \emph{reverse} of a configuration $c=(x,y,dx,dy)$ is the configuration $c^{-1}=(x+dx,y+dy,-dx,-dy)$, i.e. the same movement backwards. (One can think of two superposed vectors in opposite directions.) Thus, if $(c_1, c_2, \dots, c_k)$ is a trajectory, then $(c_k^{-1}, \dots, c_2^{-1}, c_1^{-1})$ is also a trajectory.

\subsubsection*{Configuration space.}
The set of all configurations together with the successor relation define an infinite graph called the \emph{configuration graph}. Let $\C$ be the infinite set of all configurations, the configuration graph is the directed graph $G=(V, E)$ where $V=\C$ and $E=\{(c_i,c_j) \subseteq \C \times \C : c_j \in \successors(c_i)\}$.

Basic path-finding algorithms using the configuration graph have been independently discovered many times, we consider them as folklore. Concretely, a trajectory in the original space is a path in the configuration graph. This correspondence allows for searching optimal trajectories using standard graph search like breadth-first search (BFS), in an appropriate subset of $G$ (computed on the fly).

\subsection{The \VectorTSP problem}

We introduce the \VectorTSP problem, which consists of finding a minimum length trajectory
that visits a given set of cities and returns to the starting point, starting and finishing at zero speed, and subject to \Racetrack-like physical constraints. Note that we allow to pass over a city without a visit. The order in which the cities are visited (i.e., the \emph{tour}) is not specified in the input. The goal can thus be seen as determining a \emph{tour} and a trajectory \emph{realizing} this tour.

To define the problem more formally, let us define what \emph{visiting} a city means. Here, we say that a city (or more generally, a point) $p$ is \emph{visited} by a configuration $c$ if $p$ lies on the segment between the point $p_a=pos(c) - vel(c)$ and $p_b=pos(c)$. If several cities are visited by a same configuration, then these cities are considered visited in increasing order of distances from $p_a$. 

  \begin{mdframed}
    \underline{\VectorTSP}\smallskip\\
\textbf{Input:} A set of $n$ cities (points) $P \subseteq \mathbb{Z}^d$, for some fixed $d$, a distinguished city $p_0 \in P$, and an integer $k$ encoded in unary. \smallskip\\
\textbf{Question:} Does there exist a trajectory ${\cal T}=(c_1, \dots, c_k)$ of length at most $k$ that visits all the cities in $P$, with $pos(c_1)=pos(c_k)=p_0$ and $vel(c_1)=vel(c_k)=\overset{\to}{0}$.
\end{mdframed} 

Receiving parameter $k$ in unary guarantees that the problem is not hard for artificial reasons.
For example, it prevents a trivial instance from being hard just because the cities are exponentially far from each other (in the size of the input, see also~\cite{blogernie,holzer2010computational}). 
The optimization version of \VectorTSP is defined analogously, except that instead of providing $k$ in unary, one can replace it with any number in unary that is a (polynomially related) lower bound on the optimal trajectory, such as the length of a trivial trajectory between the two more distant cities.

\subsubsection*{Stopping at the cities.}
Although the spirit of the model is that speed is unbounded, it could be useful in some scenarios to stop physically at the cities (e.g. to avoid blured pictures, or to drop objects in a delivery scenario). We define the problem \VectorTSPWS similarly as \VectorTSP, with the additional requirement that the configurations visiting a city must have zero speed. Technically, this version of the problem can be seen as somewhat intermediate between \EuclideanTSP and \VectorTSP. We do not investigate it \emph{per se} in the present paper, but we use it as a corner stone in our proof of NP-hardness of \VectorTSP in Section~\ref{sec:complexity}.

\subsubsection*{Tour / trajectory / solution.}
As already explained,
in \VectorTSP, a tour denotes a visit order of the cities (i.e. a permutation), and a trajectory denotes a sequence of racetrack configurations (that may realize such a tour). Given a tour $\pi$, we call \texttt{racetrack($\pi$)} the set of optimal trajectories realizing the tour $\pi$. It is important to keep in mind that the optimality of trajectories in \texttt{racetrack($\pi$)} is only relative to $\pi$, it does not imply that these trajectories are optimal solutions to the problem instance (if $\pi$ is itself suboptimal).
If no other tour than $\pi$ exists that admit a shorter realizing trajectory, then, in this case only, a trajectory in \texttt{racetrack($\pi$)} is an \emph{optimal solution} for the \VectorTSP problem.
\medskip

\section{Intermediate results and observations}
\label{sec:basic-results}

In this section, we establish a number of preliminary results that highlight the differences between \VectorTSP and \EuclideanTSP. We also present key properties of the configuration graph, used in subsequent sections.

\subsection{Differences between \EuclideanTSP and \VectorTSP}

\VectorTSP having essentially the same input as \EuclideanTSP, namely, a set of coordinates, it is tempting to see it as just a small variant of that problem. Both problems are however truly different.

To start, observe that the initial city does matter in \VectorTSP. This is easy to see on a set of aligned cities. If the vehicle starts at one of the endpoint cities, then it is able to visit all the cities and return after two phases of acceleration and deceleration, whereas if it starts in the middle, at least three such phases will be needed, resulting in a solution of higher cost despite the fact that the traveled distance is the same.

Beyond the starting point, we now show that optimal tours for \EuclideanTSP may be suboptimal for \VectorTSP.

\begin{lemma}
	\label{lem:opt_visit_order}
	There exists an instance $\mathcal{I}$ of \VectorTSP, such that the optimal tour $\pi$ for the corresponding \EuclideanTSP instance is suboptimal for $\mathcal{I}$.
\end{lemma}

\begin{proof}
	We provide a simple example in \cref{fig:opt_ETSP_vs_opt_VTSP}. On the left, a trajectory in \texttt{racetrack($\pi$)}, where $\pi$ is an optimal tour for \EuclideanTSP, starting and ending at $p_0$ (whence the final deceleration loop). Realizing this trajectory requires $22$ moves, but an optimal trajectory for \VectorTSP takes only $20$ (right picture) and corresponds to a tour that would be suboptimal for \EuclideanTSP. 
\end{proof}

\begin{figure}[h]
	\centering
	\includegraphics[width=\textwidth-.5cm]{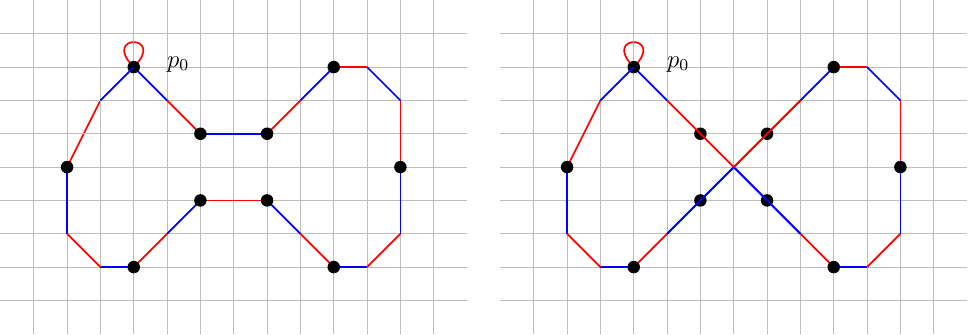}
	\caption{On the left, a trajectory in \texttt{racetrack($\pi$)}, where $\pi$ is an optimal tour for \EuclideanTSP; on the right, an optimal \VectorTSP solution using less vectors along a different tour. The loop on $p_0$ illustrates the final ``move'' required to finish at zero speed. The alternation of colors is used to help distinguish individual moves.
		\label{fig:opt_ETSP_vs_opt_VTSP}}
            \end{figure}

          From the same example, observe that optimal \VectorTSP solution may self-cross. In fact, the cost of traveling between two given cities in \VectorTSP depends crucially on choices that are made earlier in the trajectory. This feature, together with the fact that traveling some distance costs essentially the square root of that distance, makes  \VectorTSP significantly different from \EuclideanTSP and not reducible to it (at least, not straightforwardly), nor even to general \TSP (whose costs between cities are fixed).

          Later on, in this paper, we give evidence that this divergence between both problems is not occasional. For sufficiently large random instances, optimal \EuclideanTSP tours become \emph{typically} suboptimal for \VectorTSP.

\subsection{The configuration graph is (in effect) bounded}

In theory, the positions and velocities are unbounded, which makes the configuration graph infinite. Nevertheless, we show below that the \emph{useful part} of the configuration graph is bounded. Moreover, an optimal trajectory can always be found within a subgraph of the configuration graph whose size is polynomial in the size of the input.

\begin{lemma}[Length of an optimal solution]
  \label{lem:walk}
  Let $\cal I$ be a \VectorTSP instance based on a set $P$ of $n$ cities, let $L$ be the maximum distance in a dimension between two cities of~$P$ (euclidean max norm).
  An optimal solution contains at most $O(nL)$ configurations.
\end{lemma}

\begin{proof}
  A naive (suboptimal) solution consists of visiting the cities in an arbitrary order at \emph{unit speed} (i.e. the max norm of the speed is one). Such a trajectory contains $O(nL)$ configurations, which is thus an upper bounds on the size of an optimal solution.
  \end{proof}

  \Cref{lem:walk} implies that the useful part of the configuration graph can be restricted to a subgraph of \emph{polynomial} size in the size of the input.
  
\begin{lemma}[Bounds on the configuration graph]
  \label{lem:bounded-configuration-graph}
  Let $\cal I$ be a ($d$-dimensional) \VectorTSP instance based on a set $P$ of $n$ cities and a starting city $p_0 \in P$. Let $L$ be the maximum distance between two cities of~$P$ in a dimension (euclidean max norm). An optimal trajectory for $\cal I$ can be found in a subgraph of $G$ that contains at most $O((nL)^{3d})$ configurations.
\end{lemma}

\begin{proof}
  For simplicity, consider a renormalized instance in which the coordinate of $p_0$ is $0$ in each dimension. By \Cref{lem:walk}, an optimal solution has at most $O(nL)$ configurations, thus the maximum absolute value of a position in any dimension is at most $O((nL)^2)$, which is attained, for example, if one accelerates all the time in this dimension. Similarly, the maximum absolute value of the velocity in any dimension is $O(nL)$, which implies at most $O((nL)^3)$ values for the combined position and velocity in each dimension, and thus at most $O((nL)^{3d})$ values in $d$ dimensions.
\end{proof}

\begin{corollary}
  \label{lem:bounded-configuration-graph-2d}
  In two dimensions, an optimal trajectory can be found in a subgraph of the configuration graph containing $O((nL)^{6})$ vertices and edges.
\end{corollary}

\subsection{One-dimensional cost computation}
\label{subsec:unidimensional}

Some of the computation in the subsequent sections rely on computing the cost of trajectories in one dimension. Here, we first prove a direct formula for expressing the cost of moving from one position to another in one dimension, starting and finishing at zero speed. Then, we explain how more general one-dimensional computation reduces to it. 

\begin{lemma}
	In one dimension, the cost of covering $k$ space units, starting and finishing at zero speed is $\lceil 2\sqrt{k} \rceil$.
	\label{lem:onedimensioncost}
\end{lemma}

\begin{proof}
	Starting at zero speed, a vehicle accelerating over $s$ steps covers a distance of $1+2+\dots+s = s(s+1)/2$ units. Symmetrically, a vehicle decelerating down to zero speed over $s$ steps also covers $s(s+1)/2$ units. Thus, if $k=s(s+1)$ for some integer~$s$, a strategy is to accelerate for~$s$ steps and then decelerate for $s$ steps, which maximizes the average speed and is thus optimal (see the top scenario in~\Cref{fig:min}). Together with the final loop for reaching zero speed, this corresponds to a total cost of $2s + 1$. And since $s^2 < k < (s + \tfrac{1}{2})^2$, this corresponds to a cost of $\lceil 2 \sqrt{k} \rceil$.
	\medskip
	
	\noindent
	If $k$ cannot be expressed as the product of two consecutive integers, then let $s$ be the largest integer such that $s(s+1) < k$. There are two cases:
	\begin{enumerate}
		\item If $s(s+1) < k \leq (s+1)^2$, then
		$s^2 + s + 1 \leq k \leq s^2 + 2s + 1$ (since $k$ is an
		integer). In this case, there is a non-empty gap in the
		middle and a trajectory of cost $2s+1$ (including the final
		loop) cannot exist. However, we will show that a trajectory
		of cost $2s+2$ always exist in this case. Let $k' \leq s+1$
		be the size of the gap ($k'=1$ in \Cref{fig:min}, middle
		scenario). One can complete the trajectory by adding a
		single additional move of length $k'$. If $k' \geq s-1$, this
		move can be inserted in the middle. If $k' < s - 1$, then a
		move of length $k'$ already exists in the deceleration part.
		It suffices to
		repeat this move and shift the previous ones accordingly
		(dashed green vector in~\cref{fig:min}, bottom scenario) in
		order to obtain a trajectory of cost $2s+2$. Since
		$(s + 0.5)^2 < k \leq (s+1)^2$, this cost amounts again to
		$\lceil 2\sqrt{k} \rceil$.
		\item If $(s+1)^2 < k < (s+1)(s+2)$, then
		$s^2 + 2s + 2 \leq k \leq s^2 + 3s + 1$ (since $k$ is an
		integer). In this case, there is a gap of $k'$ units, with
		$s + 2 \leq k' \leq 2s + 1$. No trajectory with cost $2s+2$
		exists, since $k'$ exceeds the size allowed for an
		additional vector. However, one can complete the trajectory
		by splitting $k'$ in a sum of two numbers $\leq s+1$ and
		insert two additional such moves in the decelerating
		segment. Since $(s + 1)^2 < k \leq (s+1.5)^2$, this amounts
		to a cost of $2s + 3 = \lceil 2\sqrt{k} \rceil$.
	\end{enumerate}
\end{proof}

\begin{figure}[h]
	\begin{center}
		\begin{tikzpicture}
		[scale=.8, inner sep=1mm,
		cherry/.style={circle,draw=black,fill=red},
		blueberry/.style={circle,draw=black,fill=blue}]
		\draw[step=1cm,black!50] (1,-0.2) grid (13,0.2);
		
		\path (1,0) node[circle, fill = white, draw=black, inner sep= .7mm] (X0) {\footnotesize$x$};
		\path (13,0) node[circle, fill = white, draw=black, inner sep= .4mm] (X4) {\footnotesize$p$};
		\draw[red, ->, thick] (X0) to[bend left=70] (1.9,0);
		\draw[red, ->, thick] (2,0) to[bend left=80] (3.9,0);
		\draw[red, ->, thick] (4,0) to[bend left=80] (6.9,0);
		\draw[blue, <-, thick] (X4) to[bend right=70, out=-40] (12.1,0);
		\draw[blue, <-, thick] (12,0) to[bend right=80] (10.1,0);
		\draw[blue, <-, thick] (10,0) to[bend right=80] (7.1,0);
		
		\draw[dotted, thick] (7,1) to (7,-0.25);
		\draw[blue, ->, thick] (X4) to[out = 110, in = 180] (13, .8) to[out = 0, in = 70] (X4);			
		\end{tikzpicture}\medskip\\
		\begin{tikzpicture}
		[scale=.8, inner sep=1mm,
		cherry/.style={circle,draw=black,fill=red},
		blueberry/.style={circle,draw=black,fill=blue}]
		\draw[step=1cm,black!50] (1,-0.2) grid (14,0.2);
		
		\path (1,0) node[circle, fill = white, draw=black, inner sep= .7mm] (X0) {\footnotesize$x$};
		\path (14,0) node[circle, fill = white, draw=black, inner sep= .4mm] (X4) {\footnotesize$p$};
		\draw[red, ->, thick] (X0) to[bend left=70] (1.9,0);
		\draw[red, ->, thick] (2,0) to[bend left=80] (3.9,0);
		\draw[red, ->, thick] (4,0) to[bend left=80] (6.9,0);
		\draw[blue, <-, thick] (X4) to[bend right=70, out=-40] (13.1,0);
		\draw[blue, <-, thick] (13,0) to[bend right=80] (11.1,0);
		\draw[blue, <-, thick] (11,0) to[bend right=80] (8.1,0);
		
		\draw[dotted, thick] (7.5,1) to (7.5,-0.25);
		\draw[blue, ->, thick] (X4) to[out = 110, in = 180] (14, .8) to[out = 0, in = 70] (X4);			
		\end{tikzpicture}\\
		\begin{tikzpicture}
		[scale=.8, inner sep=1mm,
		cherry/.style={circle,draw=black,fill=red,thick},
		blueberry/.style={circle,draw=black,fill=blue,thick},
		every loop/.style={min distance=8mm,in=60,out=120,looseness=10}]
		\draw[step=1cm,black!50] (1,-0.2) grid (14,0.2);
		
		\path (1,0) node[circle, fill = white, draw=black, inner sep= .7mm] (a) {\footnotesize$x$};
		\path (14,0) node[circle, fill = white, draw=black, inner sep= .4mm] (bb) {\footnotesize$p$};
		
		\draw[red, ->, thick] (X0) to[bend left=70] (1.9,0);
		\draw[black!40!green, dashed, ->, very thick] (12,0) to [bend left=70] (12.9,0);
		\draw[red, ->, thick] (2,0) to[bend left=80] (3.9,0);
		\draw[red, ->, thick] (4,0) to[bend left=80] (6.9,0);
		\draw[blue, ->, thick] (13,0) to[bend left=70, in=140] (bb);
		\draw[blue, ->, thick] (10,0) to[bend left=80] (11.9,0);
		\draw[blue, ->, thick] (7,0) to[bend left=80] (9.9,0);
		
		\draw[blue, ->, thick] (bb) to[out = 110, in = 180] (14, .8) to[out = 0, in = 70] (bb);
		\end{tikzpicture}
	\end{center}
	\caption{Computation of an optimal (one dimensional) trajectory for $k = 12$.}
	\label{fig:min}
\end{figure}

Observe, from the previous lemma, that an optimal trajectory between two points in one dimension, starting and finishing at zero speed, can always be realized by first accelerating as much as possible without overpassing the middle point, then filling the potential gaps upon deceleration.
Another case of interest, used in~\Cref{subsec:A-star}, is when the vehicle (with position $x$ and speed $dx$) is moving towards point $x'$, and it is able to stop before bypassing $x'$. The strategy in this case is to simulate that the vehicle has accelerated continuously $dx$ time steps from a virtual starting point located at $x''=x - dx(dx+1)/2$ and to compute the cost from $x''$ to $x'$ starting and finishing at zero speed. By the above remark, the resulting cost (minus the $dx$ virtual steps) must correspond to the optimal cost towards $x'$ from the current position and speed. 

Also remark that for a pair of cities, placed in any number of dimensions, the cost when starting and stopping at zero speed, is $\lceil 2\sqrt{K} \rceil$ where $K$ is the largest distance among all dimensions.

\section{NP-hardness of \VectorTSP}
\label{sec:complexity}

Although the result itself is not surprising, proving that \VectorTSP is NP-hard does not appear to be straightforward. In this section, we do so in two steps. First, we prove that \VectorTSPWS is NP-hard (\Cref{theorem:VTSP0_nphard}), by adapting some ideas from Papadimitriou~\cite{papadimitriou1977euclidean} in his proof that \EuclideanTSP is NP-hard, via a reduction from \ExactCover. The main conceptual difficulty in this adaptation is to account for the fact that acceleration allows one to travel quadratic distance in linear time, and counter rounding effects by a scaling argument in continuous space (\Cref{theorem:CVTSP0_nphard}). Then, we show that \VectorTSPWS reduces to \VectorTSP (\Cref{theorem:VTSPinf_nphard}).

\subsection{NP-hardness of \VectorTSPWS}
\label{subsec:NPhard_VTSP0}

We choose to do most of the reasoning in this section on a slightly modified version of \VectorTSPWS, referred to as \ContinuousVectorTSPWithStops. 

  \begin{mdframed}
    \underline{\ContinuousVectorTSPWithStops}\smallskip\\
		\textbf{Input:} A set of $n$ cities (points) $P \subseteq \mathbb{Z}^d$, with intercity cost $c(p_i, p_j) = 2\sqrt{K}$ where $K$ is the largest distance between the cities among all dimensions, and an integer $k$. \smallskip\\
\textbf{Question:} Does there exist a trajectory ${\cal T}=(c_1, \dots, c_k)$ of cost at most $k$ that visits all the cities in $P$, with $pos(c_1)=pos(c_k)=p_0$ and $vel(c_1)=vel(c_k)=\overset{\to}{0}$.
	\end{mdframed} 

\ContinuousVectorTSPWithStops allows us to avoid some technical details of \VectorTSPWS, such as the rounding effects on costs between cities. These effects are taken care of later through a scaling argument which transfers NP-hardness results from \ContinuousVectorTSPWithStops to \VectorTSPWS. This version is closer in spirit to the \EuclideanTSP, allowing us to adapt Papadimitriou's NP-hardness proof for \EuclideanTSP~\cite{papadimitriou1977euclidean}. Note that we can talk about paths for \ContinuousVectorTSPWithStops instead of trajectories (as there is just the notion of a cost between cities, not of rounds/vectors). In the case of \VectorTSPWS, either paths or trajectories can be considered.

The reduction is from \ExactCover, whose definition is as follows. 
Let ${\cal U}$ be a set of $m$ elements (the {\em universe}), \ExactCover takes as input a set ${\cal F}=\{F_i\}$ of $n$ subsets of ${\cal U}$, and asks if there exists ${\cal F'} \subseteq {\cal F}$ such that all sets in ${\cal F'}$ are {\em disjoint} and ${\cal F'}$ covers all the elements of ${\cal U}$.
For example, if ${\cal U}=\{1, 2, 3\}$ and ${\cal F}=$ $\{\{1,2\}, \{3\}, \{2,3\}\}$, then ${\cal F}'= \{\{1,2\},\{3\}\}$ is a valid solution, but $\{\{1,2\},\{2,3\}\}$ is not. 
\medskip

\textbf{High-level view of the reduction:}
Given an instance \I of \ExactCover, we construct an instance \I' of \ContinuousVectorTSPWithStops such that \I admits a solution if and only if there is a trajectory visiting all the cities of \I' using at most a prescribed number of vectors.
$\I'$ is composed of several types of gadgets, representing respectively the subsets $F_i\in {\cal F}$ and the elements of ${\cal U}$ (with some repetition). For each $F_i$, a subset gadget $C_i$ is created which consists of a number of cities placed horizontally (wavy horizontal segments in Figure~\ref{fig:papa_construction}). For now, we state that each gadget can be traversed optimally in exactly two possible ways (without considering direction), which ultimately corresponds to including (mode 1) or excluding (mode 2) subset $F_i$ in the \ExactCover solution. 
The $C_{i}$'s are located one below another, starting with $C_1$ at the top. Between every two consecutive gadgets $C_i$ and $C_{i+1}$, copies of {\em element} gadgets are placed for each element in ${\cal U}$, thus the element gadgets $H_{i,j}$ are indexed by both $1 \leq i \leq n-1$ and $1 \leq j \leq m$ (see again Figure~\ref{fig:papa_construction}). The element gadgets are also made of a number of cities, whose particular organization is described later on. Finally, every subset gadget $C_i$ above or below an element gadget representing element $j$ is modified in a way that represents whether $F_i$ contains element $j$ or not.

Intuitively, a tour visiting all the cities must choose between inclusion or exclusion of each $F_i$ ({\it i.e.}, mode 1 or 2 for each $C_i$). An element $j \in {\cal U}$ is considered as covered by a subset $F_i$ if $C_i$ does {\em not} visit any of the adjacent gadgets representing $j$.
Each element gadget $H_{i,j}$ must be visited either from above (from $C_i$) or from below (from $C_{i+1}$). The number of subset gadgets is $n$, the number of element gadgets for each element is $n-1$ (one between every two consecutive subset gadgets), and the construction guarantees that at most one element gadget for each element $j\in {\cal U}$ is visited from a subset gadget $C_i$ (or the tour is suboptimal). These three properties collectively imply that for each element $j\in {\cal U}$, there is exactly one subset gadget $C_i$ that does not visit any of the element gadgets representing $j$. 

\begin{figure}[h]
	\begin{center}
		\tikzset{zigzag/.style={decorate, decoration={zigzag,amplitude=1}}}
		\begin{tikzpicture}[scale=0.5]
		\path (0, 0) node[circle, fill = white, draw=black, inner sep=.5mm] (a) {};
		\path (1, 0) node[circle, fill = white, draw=black, inner sep=.5mm] (a) {};
		\path (2, 0) node[circle, fill = white, draw=black, inner sep=.5mm] (a) {};
		\path (0, -1) node[circle, fill = white, draw=black, inner sep=.5mm] (a) {};
		\path (0, -2) node[circle, fill = white, draw=black, inner sep=.5mm] (a) {};
		\path (0, -3) node[circle, fill = white, draw=none, rotate=90, inner sep=.5mm] (a) {...};
		\path (3, 0) node[circle, fill = white, draw=none, inner sep=.5mm] (a) {...};
		
		\path (0, -4) node[circle, fill = white, draw=none, inner sep=.5mm] (below) {};
		
		\path (20, 0) node[circle, fill = white, draw=none, inner sep=0mm] (right) {};
		\path (20, -3) node[circle, fill = white, draw=none, inner sep=0mm] (right2) {};
		\path (20, -6) node[circle, fill = white, draw=none, inner sep=0mm] (right3) {};
		\path (20, -9) node[circle, fill = white, draw=none, inner sep=0mm] (right4) {};
		\path (20, -12) node[circle, fill = white, draw=none, inner sep=0mm] (right5) {};
		\path (20, -15) node[circle, fill = white, draw=none, inner sep=0mm] (right6) {};
		
		\path (4, 0) node[circle, fill = white, draw=none, inner sep=0mm] (left) {};
		\path (4, -3) node[circle, fill = none, draw=none, inner sep=0mm] (left2) {};
		\path (4, -6) node[circle, fill = white, draw=none, inner sep=0mm] (left3) {};
		\path (4, -9) node[circle, fill = white, draw=none, inner sep=0mm] (left4) {};
		\path (4, -12) node[circle, fill = white, draw=none, inner sep=.5mm] (left5) {};
		\path (4, -15) node[circle, fill = white, draw=none, inner sep=0mm] (left6) {};
		\path (0, -15) node[circle, fill = white, draw=none, inner sep=0mm] (left6b) {};
		
		\path (4, -10.5) node[circle, fill = white, draw=none, inner sep=.5mm, rotate=90] (dots) {...};
		
		\path (12, .5) node[circle, fill = white, draw=none, inner sep=.5mm] (dots) {\small $C_1$};
		\path (12, -2.5) node[circle, fill = white, draw=none, inner sep=.5mm] (dots) {\small $C_2$};
		\path (12, -5.5) node[circle, fill = white, draw=none, inner sep=.5mm] (dots) {\small $C_3$};
		\path (12, -8.5) node[circle, fill = white, draw=none, inner sep=.5mm] (dots) {\small $C_4$};
		\path (12, -11.5) node[circle, fill = white, draw=none, inner sep=.5mm] (dots) {\small $C_{n-1}$};
		\path (12, -14.5) node[circle, fill = white, draw=none, inner sep=.5mm] (dots) {\small $C_n$};

		//drawing lines
		\tikzstyle{every path}=[zigzag]
		\draw (left)  to (right);
		\draw (left2)  to (right2);
		\draw (left3)  to (right3);
		\draw (left4)  to (right4);
		\draw (right5)  to (left5);
		\draw (right6)  to (left6);
		\tikzstyle{every path}=[]
		\draw (left2)  to (left3);
		\draw (right)  to (right2);
		\draw (right3)  to (right4);
		\draw (right5)  to (right6);
		\draw (below)  to (left6b);
		\draw (left6)  to (left6b);
		
		\path (7, -1.5) node[rectangle, fill = white, draw=black, inner sep=.5mm] (dots) {\small $H_{1,1}$};
		\path (10, -1.5) node[rectangle, fill = white, draw=black, inner sep=.5mm] (dots) {\small $H_{1,2}$};
		\path (18, -1.5) node[rectangle, fill = white, draw=black, inner sep=.5mm] (dots) {\small $H_{1,m}$};
		\path (14, -1.5) node[rectangle, fill = white, draw=none, inner sep=.5mm] (dots) {\small $...$};
		
		\path (7, -4.5) node[rectangle, fill = white, draw=black, inner sep=.5mm] (dots) {\small $H_{2,1}$};
		\path (10, -4.5) node[rectangle, fill = white, draw=black, inner sep=.5mm] (dots) {\small $H_{2,2}$};
		\path (18, -4.5) node[rectangle, fill = white, draw=black, inner sep=.5mm] (dots) {\small $H_{2,m}$};
		\path (14, -4.5) node[rectangle, fill = white, draw=none, inner sep=.5mm] (dots) {\small $...$};
		
		\path (7, -7.5) node[rectangle, fill = white, draw=black, inner sep=.5mm] (dots) {\small $H_{3,1}$};
		\path (10, -7.5) node[rectangle, fill = white, draw=black, inner sep=.5mm] (dots) {\small $H_{3,2}$};
		\path (18, -7.5) node[rectangle, fill = white, draw=black, inner sep=.5mm] (dots) {\small $H_{3,m}$};
		\path (14, -7.5) node[rectangle, fill = white, draw=none, inner sep=.5mm] (dots) {\small $...$};
		
		\path (6.5, -13.5) node[rectangle, fill = white, draw=black, inner sep=.5mm] (dots) {\small $H_{n-1,1}$};
		\path (10.5, -13.5) node[rectangle, fill = white, draw=black, inner sep=.5mm] (dots) {\small $H_{n-1,2}$};
		\path (17.5, -13.5) node[rectangle, fill = white, draw=black, inner sep=.5mm] (dots) {\small $H_{n-1,m}$};
		\path (14, -13.5) node[rectangle, fill = white, draw=none, inner sep=.5mm] (dots) {$...$};
		\end{tikzpicture}
	\end{center}
	\caption{Papadimitriou's high-level construction}
	\label{fig:papa_construction}
\end{figure}

In summary, the tour proceeds from the top left corner through the $C_i$s (in order), visiting all the $H_{i,j}$. So long as a $C_i$ visits a $H_{i,j}$ (thus, from above), this means that element $j$ has not yet been covered in the \ExactCover solution. Element $j$ is covered by subset $F_i$ in the \ExactCover solution if $C_i$ is the subset gadget that does {\em not} visit the corresponding $H_{i,j}$, after which all the $H_{k \geq i,j}$ will necessarily be visited from below by the corresponding $C_{k+1}$, ensuring the element $j$ will not be covered more than once.\\

\textbf{\ConstructionOne:}
Let us now present the details of the gadgets and how to assemble them together, so as to transform any given \ExactCover instance into a \ContinuousVectorTSPWithStops instance\footnote{The cost $k$ will be missing from the instance, but we will specify this cost in \Cref{theorem:CVTSP0_nphard}.}.
Consider the \textit{1-chain} structure presented in Figure~\ref{fig:1-chain}. This structure is composed of cities positioned on a line, at distance 1 from one another. 1-chains can bend at 90 degrees angles.
Next, consider the structure in Figure~\ref{fig:2-chain}, referred to as a \textit{2-chain}. The distance between the leftmost (or rightmost) city and its nearby cities is $\frac{\sqrt{2}}{2}$. The closest distance between vertical cities is 1, and 2 between horizontal cities. \\
Finally, consider the $A$, $B$, and $H$ structures, presented in \Cref{fig:A,fig:B,fig:H} respectively. Structures $A$ and $B$ will replace parts of 2-chains as shown, and consist of cities positioned on a grid with distance 1 between nearest cities, except for the cities closest to the middle, which are horizontally distanced by 2 and 6 respectively. Structure $H$ consists of a symmetric set of cities which will be positioned far away from all other gadgets.

\begin{figure}[]
	\begin{center}
		\begin{subfigure}[b]{0.32\textwidth}
			\begin{center}
				\begin{tikzpicture}
				[scale=.5, inner sep=1mm,
				cherry/.style={circle,draw=black,fill=red,thick},
				blueberry/.style={circle,draw=black,fill=blue,thick},
				every loop/.style={min distance=8mm,in=60,out=120,looseness=10}]
				
				\path (-1,0) node[circle, fill = white, draw=black, inner sep= .7mm] (a') {};
				\path (0,0) node[circle, fill = white, draw=black, inner sep= .7mm] (a) {};
				\path (1,0) node[circle, fill = white, draw=black, inner sep= .7mm] (b) {};
				\path (5,0) node[circle, fill = white, draw=black, inner sep= .7mm] (c) {};
				\path (6,0) node[circle, fill = white, draw=black, inner sep= .7mm] (d) {};
				\path (6,-1) node[circle, fill = white, draw=black, inner sep= .7mm] (e) {};
				\path (6,-5) node[circle, fill = white, draw=black, inner sep= .7mm] (f) {};
				\path (6,-6) node[circle, fill = white, draw=black, inner sep= .7mm] (g) {};
				\path (5,-6) node[circle, fill = white, draw=black, inner sep= .7mm] (h) {};
				\path (1,-6) node[circle, fill = white, draw=black, inner sep= .7mm] (i) {};
				\path (0,-6) node[circle, fill = white, draw=black, inner sep= .7mm] (j) {};
				\path (-1,-6) node[circle, fill = white, draw=black, inner sep= .7mm] (j') {};
				
				\path (3,0) node[circle, fill = none, draw=none, inner sep= .7mm] (A) {...};
				\path (3,-6) node[circle, fill = none, draw=none, inner sep= .7mm] (A) {...};
				\path (6,-3) node[circle, fill = none, draw=none, inner sep= .7mm, rotate=90] (A) {...};
				
				\draw[<->] (0,-.7)  to (1, -.7);
				\path (.5, -1.2) node[fill = white, draw=none, inner sep=.5mm] (dots) {\footnotesize $1$};
				\end{tikzpicture}
				\caption{}
			\end{center}
		\end{subfigure}
		\begin{subfigure}[b]{0.32\textwidth}
			\begin{center}
				\begin{tikzpicture}
				[scale=.5, inner sep=1mm,
				cherry/.style={circle,draw=black,fill=red,thick},
				blueberry/.style={circle,draw=black,fill=blue,thick},
				every loop/.style={min distance=8mm,in=60,out=120,looseness=10}]
				
				\path (-1,0) node[circle, fill = white, draw=black, inner sep= .7mm] (a') {};
				\path (0,0) node[circle, fill = white, draw=black, inner sep= .7mm] (a) {};
				\path (1,0) node[circle, fill = white, draw=black, inner sep= .7mm] (b) {};
				\path (5,0) node[circle, fill = white, draw=black, inner sep= .7mm] (c) {};
				\path (6,0) node[circle, fill = white, draw=black, inner sep= .7mm] (d) {};
				\path (6,-1) node[circle, fill = white, draw=black, inner sep= .7mm] (e) {};
				\path (6,-5) node[circle, fill = white, draw=black, inner sep= .7mm] (f) {};
				\path (6,-6) node[circle, fill = white, draw=black, inner sep= .7mm] (g) {};
				\path (5,-6) node[circle, fill = white, draw=black, inner sep= .7mm] (h) {};
				\path (1,-6) node[circle, fill = white, draw=black, inner sep= .7mm] (i) {};
				\path (0,-6) node[circle, fill = white, draw=black, inner sep= .7mm] (j) {};
				\path (-1,-6) node[circle, fill = white, draw=black, inner sep= .7mm] (j') {};
				
				\draw (a') to (a);
				\draw (a) to (b);
				\draw (b) to (2,0);
				\draw (4,0) to (c);
				\draw (c) to (d);
				\draw (d) to (e);
				\draw (e) to (6,-2);
				\draw (6,-4) to (f);
				\draw (f) to (g);
				\draw (g) to (h);
				\draw (h) to (4,-6);
				\draw (2,-6) to (i);
				\draw (i) to (j);
				\draw (j) to (j');
				
				\path (3,0) node[circle, fill = none, draw=none, inner sep= .7mm] (A) {...};
				\path (3,-6) node[circle, fill = none, draw=none, inner sep= .7mm] (A) {...};
				\path (6,-3) node[circle, fill = none, draw=none, inner sep= .7mm, rotate=90] (A) {...};
				\end{tikzpicture}
				\caption{}
			\end{center}
		\end{subfigure}
		\begin{subfigure}[b]{0.32\textwidth}
			\begin{center}
				\begin{tikzpicture}
				[scale=.5, inner sep=1mm,
				cherry/.style={circle,draw=black,fill=red,thick},
				blueberry/.style={circle,draw=black,fill=blue,thick},
				every loop/.style={min distance=8mm,in=60,out=120,looseness=10}]
				
				\path (-1,0) node[circle, fill = white, draw=black, inner sep= .7mm] (a') {};
				\path (6,0) node[circle, fill = white, draw=none, inner sep= .1mm] (d) {};
				\path (6,-6) node[circle, fill = white, draw=none, inner sep= .1mm] (g) {};
				\path (-1,-6) node[circle, fill = white, draw=black, inner sep= .7mm] (j') {};
				
				\draw (a') to (d);
				\draw (d) to (g);
				\draw (g) to (j');
				\end{tikzpicture}
				\caption{}
			\end{center}
		\end{subfigure}
	\end{center}
	\caption{1-chain structure which turns $90^{\circ}$ twice. The distance between consecutive cities is~1. The optimal visit order is shown in $(b)$. We abbreviate a 1-chain schematically as shown in $(c)$.
		\label{fig:1-chain}}
	\begin{center}
		\begin{subfigure}[b]{0.49\textwidth}
			\begin{center}
				\begin{tikzpicture}
					[scale=.45, inner sep=1mm,
					cherry/.style={circle,draw=black,fill=red,thick},
					blueberry/.style={circle,draw=black,fill=blue,thick},
					every loop/.style={min distance=8mm,in=60,out=120,looseness=10}]
					
					\path (0,0) node[circle, fill = white, draw=black, inner sep= .7mm] (a) {};
					\path (1,.6) node[circle, fill = white, draw=black, inner sep= .7mm] (b) {};
					\path (1,-.6) node[circle, fill = white, draw=black, inner sep= .7mm] (c) {};
					\path (3,.6) node[circle, fill = white, draw=black, inner sep= .7mm] (d) {};
					\path (3,-.6) node[circle, fill = white, draw=black, inner sep= .7mm] (e) {};
					\path (5,.6) node[circle, fill = white, draw=black, inner sep= .7mm] (f) {};
					\path (5,-.6) node[circle, fill = white, draw=black, inner sep= .7mm] (g) {};
					
					\path (9,.6) node[circle, fill = white, draw=black, inner sep= .7mm] (h) {};
					\path (9,-.6) node[circle, fill = white, draw=black, inner sep= .7mm] (i) {};
					\path (10,0) node[circle, fill = white, draw=black, inner sep= .7mm] (j) {};
					
					\path (7,0) node[circle, fill = none, draw=none, inner sep= .7mm] (A) {...};
					
					\draw[<->] (1,1.7)  to (3, 1.7);
					\path (2, 2.2) node[fill = white, draw=none, inner sep=.5mm] (dots) {\footnotesize $2$};
				\end{tikzpicture}
				\caption{}
			\end{center}
		\end{subfigure}
		\begin{subfigure}[b]{0.49\textwidth}
			\begin{center}
				\begin{tikzpicture}
					[scale=.45, inner sep=1mm,
					cherry/.style={circle,draw=black,fill=red,thick},
					blueberry/.style={circle,draw=black,fill=blue,thick},
					every loop/.style={min distance=8mm,in=60,out=120,looseness=10}]
					
					\path (0,0) node[circle, fill = white, draw=black, inner sep= .7mm] (a) {};
					\path (1,.6) node[circle, fill = white, draw=black, inner sep= .7mm] (b) {};
					\path (1,-.6) node[circle, fill = white, draw=black, inner sep= .7mm] (c) {};
					\path (3,.6) node[circle, fill = white, draw=black, inner sep= .7mm] (d) {};
					\path (3,-.6) node[circle, fill = white, draw=black, inner sep= .7mm] (e) {};
					\path (5,.6) node[circle, fill = white, draw=black, inner sep= .7mm] (f) {};
					\path (5,-.6) node[circle, fill = white, draw=black, inner sep= .7mm] (g) {};
					
					\path (7,0) node[circle, fill = none, draw=none, inner sep= .7mm] (A) {...};
					
					\draw (a) to (b);
					\draw (b) to (c);
					\draw (c) to (e);
					\draw (e) to (d);
					\draw (d) to (f);
					\draw (f) to (g);
					\draw (g) to (6,-.6);
				\end{tikzpicture}
				\caption{}
			\end{center}
		\end{subfigure}
		
		\vspace{.5cm}
		
		\begin{subfigure}[b]{0.49\textwidth}
			\begin{center}
				\begin{tikzpicture}
					[scale=.45, inner sep=1mm,
					cherry/.style={circle,draw=black,fill=red,thick},
					blueberry/.style={circle,draw=black,fill=blue,thick},
					every loop/.style={min distance=8mm,in=60,out=120,looseness=10}]
					
					\path (0,0) node[circle, fill = white, draw=black, inner sep= .7mm] (a) {};
					\path (1,.6) node[circle, fill = white, draw=black, inner sep= .7mm] (b) {};
					\path (1,-.6) node[circle, fill = white, draw=black, inner sep= .7mm] (c) {};
					\path (3,.6) node[circle, fill = white, draw=black, inner sep= .7mm] (d) {};
					\path (3,-.6) node[circle, fill = white, draw=black, inner sep= .7mm] (e) {};
					\path (5,.6) node[circle, fill = white, draw=black, inner sep= .7mm] (f) {};
					\path (5,-.6) node[circle, fill = white, draw=black, inner sep= .7mm] (g) {};
					
					\path (7,0) node[circle, fill = none, draw=none, inner sep= .7mm] (A) {...};
					
					\draw (a) to (c);
					\draw (b) to (c);
					\draw (b) to (d);
					\draw (e) to (d);
					\draw (e) to (g);
					\draw (f) to (g);
					\draw (f) to (6,.6);
				\end{tikzpicture}
				\caption{}
			\end{center}
		\end{subfigure}
		\begin{subfigure}[b]{0.49\textwidth}
			\begin{center}
				\begin{tikzpicture}
					[scale=.45, inner sep=1mm,
					cherry/.style={circle,draw=black,fill=red,thick},
					blueberry/.style={circle,draw=black,fill=blue,thick},
					every loop/.style={min distance=8mm,in=60,out=120,looseness=10}]
					
					\path (0,0) node[circle, fill = white, draw=black, inner sep= .7mm] (a) {};
					\path (1,.6) node[circle, fill = white, draw=none, inner sep= .7mm] (b) {};
					\path (1,-.6) node[circle, fill = white, draw=none, inner sep= .7mm] (c) {};
					\path (3,.6) node[circle, fill = white, draw=none, inner sep= .7mm] (d) {};
					\path (3,-.6) node[circle, fill = white, draw=none, inner sep= .7mm] (e) {};
					\path (9,.6) node[circle, fill = white, draw=none, inner sep= .7mm] (f) {};
					\path (9,-.6) node[circle, fill = white, draw=none, inner sep= .7mm] (g) {};
					
					\path (10,0) node[circle, fill = none, draw=black, inner sep= .7mm] (A) {};
					
					\draw (b) to (f);
					\path (5,0) node[circle, fill = none, draw=none, inner sep= .7mm] (A) {\footnotesize $C_{i}$};
					\draw (c) to (g);
				\end{tikzpicture} 
				\caption{}
			\end{center}
		\end{subfigure}
		
	\end{center}
	\caption{2-chain structure $(a)$. A 2-chain has precisely two optimal 1-trajectories, mode 1 shown in $(b)$ and mode 2 shown in $(c)$. We abbreviate a 2-chain schematically as shown in $(d)$. \label{fig:2-chain}}
	\begin{center}
		\begin{subfigure}[b]{0.3\textwidth}
			\begin{center}
				\begin{tikzpicture}
					[scale=.42, inner sep=1mm,
					cherry/.style={circle,draw=black,fill=red,thick},
					blueberry/.style={circle,draw=black,fill=blue,thick},
					every loop/.style={min distance=8mm,in=60,out=120,looseness=10}]
					
					\path (0,0) node[circle, fill = white, draw=black, inner sep= .7mm] (a) {};
					\path (1,0) node[circle, fill = white, draw=black, inner sep= .7mm] (a) {};
					\path (1,1) node[circle, fill = white, draw=black, inner sep= .7mm] (a) {};
					\path (0,1) node[circle, fill = white, draw=black, inner sep= .7mm] (a) {};
					\path (0,3) node[circle, fill = white, draw=black, inner sep= .7mm] (a) {};
					\path (1,3) node[circle, fill = white, draw=black, inner sep= .7mm] (a) {};
					\path (0,4) node[circle, fill = white, draw=black, inner sep= .7mm] (a) {};
					\path (1,4) node[circle, fill = white, draw=black, inner sep= .7mm] (a) {};
					\path (0,6) node[circle, fill = white, draw=black, inner sep= .7mm] (a) {};
					\path (1,6) node[circle, fill = white, draw=black, inner sep= .7mm] (a) {};
					\path (0,7) node[circle, fill = white, draw=black, inner sep= .7mm] (a) {};
					\path (1,7) node[circle, fill = white, draw=black, inner sep= .7mm] (a) {};
					\path (7,6) node[circle, fill = white, draw=black, inner sep= .7mm] (a) {};
					\path (7,7) node[circle, fill = white, draw=black, inner sep= .7mm] (a) {};
					\path (8,6) node[circle, fill = white, draw=black, inner sep= .7mm] (a) {};
					\path (8,7) node[circle, fill = white, draw=black, inner sep= .7mm] (a) {};
					\path (7,3) node[circle, fill = white, draw=black, inner sep= .7mm] (a) {};
					\path (7,4) node[circle, fill = white, draw=black, inner sep= .7mm] (a) {};
					\path (8,3) node[circle, fill = white, draw=black, inner sep= .7mm] (a) {};
					\path (8,4) node[circle, fill = white, draw=black, inner sep= .7mm] (a) {};
					\path (7,0) node[circle, fill = white, draw=black, inner sep= .7mm] (a) {};
					\path (7,1) node[circle, fill = white, draw=black, inner sep= .7mm] (a) {};
					\path (8,0) node[circle, fill = white, draw=black, inner sep= .7mm] (a) {};
					\path (8,1) node[circle, fill = white, draw=black, inner sep= .7mm] (a) {};
					
					\path (-.5,-.5) node[circle, fill = none, draw=none, inner sep= .7mm] (A) {\footnotesize a};
					\path (8.5,-.5) node[circle, fill = none, draw=none, inner sep= .7mm] (A) {\footnotesize a'};
					\path (1.5,-.5) node[circle, fill = none, draw=none, inner sep= .7mm] (A) {\footnotesize b};
					\path (6.5,-.5) node[circle, fill = none, draw=none, inner sep= .7mm] (A) {\footnotesize b'};
					\path (-.5,7.5) node[circle, fill = none, draw=none, inner sep= .7mm] (A) {\footnotesize c};
					\path (8.5,7.5) node[circle, fill = none, draw=none, inner sep= .7mm] (A) {\footnotesize c'};
					\path (1.5,7.5) node[circle, fill = none, draw=none, inner sep= .7mm] (A) {\footnotesize d};
					\path (6.5,7.5) node[circle, fill = none, draw=none, inner sep= .7mm] (A) {\footnotesize d'};
					
					\draw[<->] (0,-.7)  to (1, -.7);
					\path (.5, -1.2) node[fill = white, draw=none, inner sep=.5mm] (dots) {\footnotesize $1$};
					\draw[<->] (-0.8,1)  to (-0.8, 3);
					\path (-1.3, 2) node[fill = white, draw=none, inner sep=.5mm] (dots) {\footnotesize $2$};
					\draw[<->] (1,5.3)  to (7,5.3);
					\path (4, 4.8) node[fill = white, draw=none, inner sep=.5mm] (dots) {\footnotesize $6$};
				\end{tikzpicture}
				\caption{}
			\end{center}
		\end{subfigure}
		\hspace{.5cm}
		\begin{subfigure}[b]{0.3\textwidth}
			\begin{center}
				\begin{tikzpicture}
					[scale=.42, inner sep=1mm,
					cherry/.style={circle,draw=black,fill=red,thick},
					blueberry/.style={circle,draw=black,fill=blue,thick},
					every loop/.style={min distance=8mm,in=60,out=120,looseness=10}]
					
					\path (0,0) node[circle, fill = white, draw=black, inner sep= .7mm] (a) {};
					\path (1,0) node[circle, fill = white, draw=black, inner sep= .7mm] (b) {};
					\path (1,1) node[circle, fill = white, draw=black, inner sep= .7mm] (c) {};
					\path (0,1) node[circle, fill = white, draw=black, inner sep= .7mm] (d) {};
					\path (0,3) node[circle, fill = white, draw=black, inner sep= .7mm] (e) {};
					\path (1,3) node[circle, fill = white, draw=black, inner sep= .7mm] (f) {};
					\path (0,4) node[circle, fill = white, draw=black, inner sep= .7mm] (g) {};
					\path (1,4) node[circle, fill = white, draw=black, inner sep= .7mm] (h) {};
					\path (0,6) node[circle, fill = white, draw=black, inner sep= .7mm] (i) {};
					\path (1,6) node[circle, fill = white, draw=black, inner sep= .7mm] (j) {};
					\path (0,7) node[circle, fill = white, draw=black, inner sep= .7mm] (k) {};
					\path (1,7) node[circle, fill = white, draw=black, inner sep= .7mm] (l) {};
					\path (7,6) node[circle, fill = white, draw=black, inner sep= .7mm] (m) {};
					\path (7,7) node[circle, fill = white, draw=black, inner sep= .7mm] (n) {};
					\path (8,6) node[circle, fill = white, draw=black, inner sep= .7mm] (o) {};
					\path (8,7) node[circle, fill = white, draw=black, inner sep= .7mm] (p) {};
					\path (7,3) node[circle, fill = white, draw=black, inner sep= .7mm] (q) {};
					\path (7,4) node[circle, fill = white, draw=black, inner sep= .7mm] (r) {};
					\path (8,3) node[circle, fill = white, draw=black, inner sep= .7mm] (s) {};
					\path (8,4) node[circle, fill = white, draw=black, inner sep= .7mm] (t) {};
					\path (7,0) node[circle, fill = white, draw=black, inner sep= .7mm] (u) {};
					\path (7,1) node[circle, fill = white, draw=black, inner sep= .7mm] (v) {};
					\path (8,0) node[circle, fill = white, draw=black, inner sep= .7mm] (w) {};
					\path (8,1) node[circle, fill = white, draw=black, inner sep= .7mm] (x) {};
					
					\draw[thick] (a) to (b);
					\draw[thick] (b) to (c);
					\draw[thick] (c) to (d);
					\draw[thick] (d) to (e);
					\draw[thick] (e) to (f);
					\draw[thick] (f) to (h);
					\draw[thick] (h) to (g);
					\draw[thick] (g) to (i);
					\draw[thick] (i) to (k);
					\draw[thick] (k) to (l);
					\draw[thick] (l) to (j);
					\draw[thick] (j) to (m);
					\draw[thick] (m) to (n);
					\draw[thick] (o) to (p);
					\draw[thick] (p) to (n);
					\draw[thick] (o) to (t);
					\draw[thick] (t) to (r);
					\draw[thick] (r) to (q);
					\draw[thick] (q) to (s);
					\draw[thick] (s) to (x);
					\draw[thick] (x) to (v);
					\draw[thick] (v) to (u);
					\draw[thick] (u) to (w);
					
					\path (-.5,-.5) node[circle, fill = none, draw=none, inner sep= .7mm] (A) {\footnotesize a};
					\path (8.5,-.5) node[circle, fill = none, draw=none, inner sep= .7mm] (A) {\footnotesize a'};
					\path (1.5,-.5) node[circle, fill = none, draw=none, inner sep= .7mm] (A) {\footnotesize b};
					\path (6.5,-.5) node[circle, fill = none, draw=none, inner sep= .7mm] (A) {\footnotesize b'};
					\path (-.5,7.5) node[circle, fill = none, draw=none, inner sep= .7mm] (A) {\footnotesize c};
					\path (8.5,7.5) node[circle, fill = none, draw=none, inner sep= .7mm] (A) {\footnotesize c'};
					\path (1.5,7.5) node[circle, fill = none, draw=none, inner sep= .7mm] (A) {\footnotesize d};
					\path (6.5,7.5) node[circle, fill = none, draw=none, inner sep= .7mm] (A) {\footnotesize d'};
				\end{tikzpicture}
				\caption{}
			\end{center}
		\end{subfigure}
		\begin{subfigure}[b]{0.3\textwidth}
			\begin{center}
				\begin{tikzpicture}
					[scale=.3, inner sep=1mm,
					cherry/.style={circle,draw=black,fill=red,thick},
					blueberry/.style={circle,draw=black,fill=blue,thick},
					every loop/.style={min distance=8mm,in=60,out=120,looseness=10}]
					
					\path (4,4) node[rectangle, fill = none, draw=black, inner sep= .7cm] (A) { $H_{i,j}$};
					
				\end{tikzpicture}
				\caption{}
			\end{center}
		\end{subfigure}
	\end{center}
	\caption{Structure $H$. 
		An optimal 1-trajectory in $H$ is shown in $(b)$. We abbreviate an $H$ structure schematically as shown in $(c)$.		
		\label{fig:H}}
\end{figure}

\begin{figure}[]
	\begin{center}
		\begin{subfigure}[b]{0.49\textwidth}
			\begin{center}
				\begin{tikzpicture}
				[scale=.25, inner sep=1mm,
				cherry/.style={circle,draw=black,fill=red,thick},
				blueberry/.style={circle,draw=black,fill=blue,thick},
				every loop/.style={min distance=8mm,in=60,out=120,looseness=10}]
				
				\path (-1,1) node[circle, fill = white, draw=black, inner sep= .7mm] (b) {};
				\path (-1,-1) node[circle, fill = white, draw=black, inner sep= .7mm] (c) {};
				\path (3,1) node[circle, fill = white, draw=black, inner sep= .7mm] (d) {};
				\path (3,-1) node[circle, fill = white, draw=black, inner sep= .7mm] (e) {};
				\path (5,1) node[circle, fill = white, draw=black, inner sep= .7mm] (f) {};
				\path (5,-1) node[circle, fill = white, draw=black, inner sep= .7mm] (g) {};
				\path (7,1) node[circle, fill = white, draw=black, inner sep= .7mm] (h) {};
				\path (7,-1) node[circle, fill = white, draw=black, inner sep= .7mm] (i) {};
				\path (9,1) node[circle, fill = white, draw=black, inner sep= .7mm] (j) {};
				\path (9,3) node[circle, fill = white, draw=black, inner sep= .7mm] (k) {};
				\path (7,3) node[circle, fill = white, draw=black, inner sep= .7mm] (l) {};
				\path (5,3) node[circle, fill = white, draw=black, inner sep= .7mm] (m) {};
				\path (5,-3) node[circle, fill = white, draw=black, inner sep= .7mm] (n) {};
				\path (5,-5) node[circle, fill = white, draw=black, inner sep= .7mm] (o) {};
				\path (7,-3) node[circle, fill = white, draw=black, inner sep= .7mm] (p) {};
				\path (7,-5) node[circle, fill = white, draw=black, inner sep= .7mm] (q) {};
				
				\path (-2,0) node[circle, fill = none, draw=none, inner sep= .7mm] (dots) {...};
				
				\path (23,1) node[circle, fill = white, draw=black, inner sep= .7mm] (b) {};
				\path (23,-1) node[circle, fill = white, draw=black, inner sep= .7mm] (c) {};
				\path (19,1) node[circle, fill = white, draw=black, inner sep= .7mm] (d) {};
				\path (19,-1) node[circle, fill = white, draw=black, inner sep= .7mm] (e) {};
				\path (17,1) node[circle, fill = white, draw=black, inner sep= .7mm] (f) {};
				\path (17,-1) node[circle, fill = white, draw=black, inner sep= .7mm] (g) {};
				\path (15,1) node[circle, fill = white, draw=black, inner sep= .7mm] (h) {};
				\path (15,-1) node[circle, fill = white, draw=black, inner sep= .7mm] (i) {};
				\path (13,1) node[circle, fill = white, draw=black, inner sep= .7mm] (j) {};
				\path (13,3) node[circle, fill = white, draw=black, inner sep= .7mm] (k) {};
				\path (15,3) node[circle, fill = white, draw=black, inner sep= .7mm] (l) {};
				\path (17,3) node[circle, fill = white, draw=black, inner sep= .7mm] (m) {};
				\path (17,-3) node[circle, fill = white, draw=black, inner sep= .7mm] (n) {};
				\path (17,-5) node[circle, fill = white, draw=black, inner sep= .7mm] (o) {};
				\path (15,-3) node[circle, fill = white, draw=black, inner sep= .7mm] (p) {};
				\path (15,-5) node[circle, fill = white, draw=black, inner sep= .7mm] (q) {};
				
				\path (24,0) node[circle, fill = none, draw=none, inner sep= .7mm] (dots) {...};
				
				\draw[dashed] (2, -6)  to (2, 4) to (20, 4) to (20,-6) to (2, -6);
				\draw[<->] (9,0)  to (13, 0);
				\path (11, -.8) node[fill = white, draw=none, inner sep=.5mm] (dots) {\footnotesize $2$};
				\end{tikzpicture}
				\caption{}
			\end{center}
		\end{subfigure}
	\end{center}
	
	\begin{subfigure}[b]{0.49\textwidth}
		\begin{center}
			\begin{tikzpicture}
			[scale=.25, inner sep=1mm,
			cherry/.style={circle,draw=black,fill=red,thick},
			blueberry/.style={circle,draw=black,fill=blue,thick},
			every loop/.style={min distance=8mm,in=60,out=120,looseness=10}]
			
			\path (-1,1) node[circle, fill = white, draw=black, inner sep= .7mm] (b) {};
			\path (-1,-1) node[circle, fill = white, draw=black, inner sep= .7mm] (c) {};
			\path (3,1) node[circle, fill = white, draw=black, inner sep= .7mm] (d) {};
			\path (3,-1) node[circle, fill = white, draw=black, inner sep= .7mm] (e) {};
			\path (5,1) node[circle, fill = white, draw=black, inner sep= .7mm] (f) {};
			\path (5,-1) node[circle, fill = white, draw=black, inner sep= .7mm] (g) {};
			\path (7,1) node[circle, fill = white, draw=black, inner sep= .7mm] (h) {};
			\path (7,-1) node[circle, fill = white, draw=black, inner sep= .7mm] (i) {};
			\path (9,1) node[circle, fill = white, draw=black, inner sep= .7mm] (j) {};
			\path (9,3) node[circle, fill = white, draw=black, inner sep= .7mm] (k) {};
			\path (7,3) node[circle, fill = white, draw=black, inner sep= .7mm] (l) {};
			\path (5,3) node[circle, fill = white, draw=black, inner sep= .7mm] (m) {};
			\path (5,-3) node[circle, fill = white, draw=black, inner sep= .7mm] (n) {};
			\path (5,-5) node[circle, fill = white, draw=black, inner sep= .7mm] (o) {};
			\path (7,-3) node[circle, fill = white, draw=black, inner sep= .7mm] (p) {};
			\path (7,-5) node[circle, fill = white, draw=black, inner sep= .7mm] (q) {};
			
			\path (-2,0) node[circle, fill = none, draw=none, inner sep= .7mm] (dots) {...};
			
			\draw[] (-2,-1) to (c)  to (b) to (d) to (e) to (g) to (n) to (o) to (q) to (p) to (i) to (h) to (f) to (m) to (l) to (k) to (j);

			\path (13,1) node[circle, fill = white, draw=black, inner sep= .7mm] (j') {};
			\draw[] (j) to (j'); 
			
			\path (23,1) node[circle, fill = white, draw=black, inner sep= .7mm] (b) {};
			\path (23,-1) node[circle, fill = white, draw=black, inner sep= .7mm] (c) {};
			\path (19,1) node[circle, fill = white, draw=black, inner sep= .7mm] (d) {};
			\path (19,-1) node[circle, fill = white, draw=black, inner sep= .7mm] (e) {};
			\path (17,1) node[circle, fill = white, draw=black, inner sep= .7mm] (f) {};
			\path (17,-1) node[circle, fill = white, draw=black, inner sep= .7mm] (g) {};
			\path (15,1) node[circle, fill = white, draw=black, inner sep= .7mm] (h) {};
			\path (15,-1) node[circle, fill = white, draw=black, inner sep= .7mm] (i) {};
			\path (13,1) node[circle, fill = white, draw=black, inner sep= .7mm] (j) {};
			\path (13,3) node[circle, fill = white, draw=black, inner sep= .7mm] (k) {};
			\path (15,3) node[circle, fill = white, draw=black, inner sep= .7mm] (l) {};
			\path (17,3) node[circle, fill = white, draw=black, inner sep= .7mm] (m) {};
			\path (17,-3) node[circle, fill = white, draw=black, inner sep= .7mm] (n) {};
			\path (17,-5) node[circle, fill = white, draw=black, inner sep= .7mm] (o) {};
			\path (15,-3) node[circle, fill = white, draw=black, inner sep= .7mm] (p) {};
			\path (15,-5) node[circle, fill = white, draw=black, inner sep= .7mm] (q) {};
			
			\draw[] (24,-1) to (c)  to (b) to (d) to (e) to (g) to (n) to (o) to (q) to (p) to (i) to (h) to (f) to (m) to (l) to (k) to (j);
			
			\path (24,0) node[circle, fill = none, draw=none, inner sep= .7mm] (dots) {...};
			
			\draw[dashed] (2, -6)  to (2, 4) to (20, 4) to (20,-6) to (2, -6);
			\end{tikzpicture}
			\caption{}
			\vspace{10pt}
		\end{center}
	\end{subfigure}
	\begin{subfigure}[b]{0.49\textwidth}
		\begin{center}
			\begin{tikzpicture}
			[scale=.25, inner sep=1mm,
			cherry/.style={circle,draw=black,fill=red,thick},
			blueberry/.style={circle,draw=black,fill=blue,thick},
			every loop/.style={min distance=8mm,in=60,out=120,looseness=10}]
			
			\path (-1,1) node[circle, fill = white, draw=black, inner sep= .7mm] (b) {};
			\path (-1,-1) node[circle, fill = white, draw=black, inner sep= .7mm] (c) {};
			\path (3,1) node[circle, fill = white, draw=black, inner sep= .7mm] (d) {};
			\path (3,-1) node[circle, fill = white, draw=black, inner sep= .7mm] (e) {};
			\path (5,1) node[circle, fill = white, draw=black, inner sep= .7mm] (f) {};
			\path (5,-1) node[circle, fill = white, draw=black, inner sep= .7mm] (g) {};
			\path (7,1) node[circle, fill = white, draw=black, inner sep= .7mm] (h) {};
			\path (7,-1) node[circle, fill = white, draw=black, inner sep= .7mm] (i) {};
			\path (9,1) node[circle, fill = white, draw=black, inner sep= .7mm] (j) {};
			\path (9,3) node[circle, fill = white, draw=black, inner sep= .7mm] (k) {};
			\path (7,3) node[circle, fill = white, draw=black, inner sep= .7mm] (l) {};
			\path (5,3) node[circle, fill = white, draw=black, inner sep= .7mm] (m) {};
			\path (5,-3) node[circle, fill = white, draw=black, inner sep= .7mm] (n) {};
			\path (5,-5) node[circle, fill = white, draw=black, inner sep= .7mm] (o) {};
			\path (7,-3) node[circle, fill = white, draw=black, inner sep= .7mm] (p) {};
			\path (7,-5) node[circle, fill = white, draw=black, inner sep= .7mm] (q) {};
			
			\path (-2,0) node[circle, fill = none, draw=none, inner sep= .7mm] (dots) {...};
			
			\draw[] (-2,1) to (b)  to (c) to (e) to (d) to (f) to (g) to (n) to (o) to (q) to (p) to (i) to (h) to (j) to (k) to (k) to (l) to (m) to (5, 5);
			\path (5,6) node[circle, fill = none, draw=none, inner sep= .7mm, rotate=90] (dots) {...};
			
			\path (23,1) node[circle, fill = white, draw=black, inner sep= .7mm] (b) {};
			\path (23,-1) node[circle, fill = white, draw=black, inner sep= .7mm] (c) {};
			\path (19,1) node[circle, fill = white, draw=black, inner sep= .7mm] (d) {};
			\path (19,-1) node[circle, fill = white, draw=black, inner sep= .7mm] (e) {};
			\path (17,1) node[circle, fill = white, draw=black, inner sep= .7mm] (f) {};
			\path (17,-1) node[circle, fill = white, draw=black, inner sep= .7mm] (g) {};
			\path (15,1) node[circle, fill = white, draw=black, inner sep= .7mm] (h) {};
			\path (15,-1) node[circle, fill = white, draw=black, inner sep= .7mm] (i) {};
			\path (13,1) node[circle, fill = white, draw=black, inner sep= .7mm] (j) {};
			\path (13,3) node[circle, fill = white, draw=black, inner sep= .7mm] (k) {};
			\path (15,3) node[circle, fill = white, draw=black, inner sep= .7mm] (l) {};
			\path (17,3) node[circle, fill = white, draw=black, inner sep= .7mm] (m) {};
			\path (17,-3) node[circle, fill = white, draw=black, inner sep= .7mm] (n) {};
			\path (17,-5) node[circle, fill = white, draw=black, inner sep= .7mm] (o) {};
			\path (15,-3) node[circle, fill = white, draw=black, inner sep= .7mm] (p) {};
			\path (15,-5) node[circle, fill = white, draw=black, inner sep= .7mm] (q) {};
			
			\draw[] (24,1) to (b)  to (c) to (e) to (d) to (f) to (g) to (n) to (o) to (q) to (p) to (i) to (h) to (j) to (k) to (k) to (l) to (m) to (17, 5);
			\path (17,6) node[circle, fill = none, draw=none, inner sep= .7mm, rotate=90] (dots) {...};

			\path (24,0) node[circle, fill = none, draw=none, inner sep= .7mm] (dots) {...};

			\draw[dashed] (2, -6)  to (2, 4) to (20, 4) to (20,-6) to (2, -6);
			\end{tikzpicture}
			\caption{}
		\end{center}
		\vspace{2pt}
	\end{subfigure}
	\caption{Structure $A$ with parts of a 2-chain on the sides (see $(a)$). Visiting structure $A$ in mode 1 makes it suboptimal to visit an $H$ structure above or below (see $(b)$). Visiting structure $A$ in mode 2 however, makes it less costly to visit an $H$ structure above or below (see for example $(c)$ for the former).
		\label{fig:A}}
	\begin{center}
		\begin{subfigure}[b]{0.49\textwidth}
			\begin{center}
				\begin{tikzpicture}
					[scale=.25, inner sep=1mm,
					cherry/.style={circle,draw=black,fill=red,thick},
					blueberry/.style={circle,draw=black,fill=blue,thick},
					every loop/.style={min distance=8mm,in=60,out=120,looseness=10}]
					
					\path (-1,1) node[circle, fill = white, draw=black, inner sep= .7mm] (b) {};
					\path (-1,-1) node[circle, fill = white, draw=black, inner sep= .7mm] (c) {};
					\path (3,1) node[circle, fill = white, draw=black, inner sep= .7mm] (d) {};
					\path (3,-1) node[circle, fill = white, draw=black, inner sep= .7mm] (e) {};
					\path (5,1) node[circle, fill = white, draw=black, inner sep= .7mm] (f) {};
					\path (5,-1) node[circle, fill = white, draw=black, inner sep= .7mm] (g) {};
					\path (3,3) node[circle, fill = white, draw=black, inner sep= .7mm] (h) {};
					\path (3,-3) node[circle, fill = white, draw=black, inner sep= .7mm] (i) {};
					\path (3,-5) node[circle, fill = white, draw=black, inner sep= .7mm] (j) {};
					\path (5,3) node[circle, fill = white, draw=black, inner sep= .7mm] (m) {};
					\path (5,-3) node[circle, fill = white, draw=black, inner sep= .7mm] (n) {};
					\path (5,-5) node[circle, fill = white, draw=black, inner sep= .7mm] (o) {};
					
					\path (-2,0) node[circle, fill = none, draw=none, inner sep= .7mm] (dots) {...};
					
					\path (23,1) node[circle, fill = white, draw=black, inner sep= .7mm] (b) {};
					\path (23,-1) node[circle, fill = white, draw=black, inner sep= .7mm] (c) {};
					\path (19,1) node[circle, fill = white, draw=black, inner sep= .7mm] (d) {};
					\path (19,-1) node[circle, fill = white, draw=black, inner sep= .7mm] (e) {};
					\path (17,1) node[circle, fill = white, draw=black, inner sep= .7mm] (f) {};
					\path (17,-1) node[circle, fill = white, draw=black, inner sep= .7mm] (g) {};
					\path (19,3) node[circle, fill = white, draw=black, inner sep= .7mm] (h) {};
					\path (19,-3) node[circle, fill = white, draw=black, inner sep= .7mm] (i) {};
					\path (19,-5) node[circle, fill = white, draw=black, inner sep= .7mm] (j) {};
					\path (17,3) node[circle, fill = white, draw=black, inner sep= .7mm] (m) {};
					\path (17,-3) node[circle, fill = white, draw=black, inner sep= .7mm] (n) {};
					\path (17,-5) node[circle, fill = white, draw=black, inner sep= .7mm] (o) {};
					
					\path (24,0) node[circle, fill = none, draw=none, inner sep= .7mm] (dots) {...};
					
					\draw[dashed] (2, -6)  to (2, 4) to (20, 4) to (20,-6) to (2, -6);
					\draw[<->] (5.7,1)  to (16.3, 1);
					\path (11, .1) node[fill = white, draw=none, inner sep=.5mm] (dots) {\footnotesize $6$};
				\end{tikzpicture}
				\caption{}
			\end{center}
		\end{subfigure}
	\end{center}
	
	\begin{subfigure}[b]{0.49\textwidth}
		\begin{center}
			\begin{tikzpicture}
				[scale=.25, inner sep=1mm,
				cherry/.style={circle,draw=black,fill=red,thick},
				blueberry/.style={circle,draw=black,fill=blue,thick},
				every loop/.style={min distance=8mm,in=60,out=120,looseness=10}]
				
				\path (-1,1) node[circle, fill = white, draw=black, inner sep= .7mm] (b) {};
				\path (-1,-1) node[circle, fill = white, draw=black, inner sep= .7mm] (c) {};
				\path (3,1) node[circle, fill = white, draw=black, inner sep= .7mm] (d) {};
				\path (3,-1) node[circle, fill = white, draw=black, inner sep= .7mm] (e) {};
				\path (5,1) node[circle, fill = white, draw=black, inner sep= .7mm] (f) {};
				\path (5,-1) node[circle, fill = white, draw=black, inner sep= .7mm] (g) {};
				\path (3,3) node[circle, fill = white, draw=black, inner sep= .7mm] (h) {};
				\path (3,-3) node[circle, fill = white, draw=black, inner sep= .7mm] (i) {};
				\path (3,-5) node[circle, fill = white, draw=black, inner sep= .7mm] (j) {};
				\path (5,3) node[circle, fill = white, draw=black, inner sep= .7mm] (m) {};
				\path (5,-3) node[circle, fill = white, draw=black, inner sep= .7mm] (n) {};
				\path (5,-5) node[circle, fill = white, draw=black, inner sep= .7mm] (o) {};
				
				\path (-2,0) node[circle, fill = none, draw=none, inner sep= .7mm] (dots) {...};
				
				\draw[] (-2, -1) to (c)  to (b) to (d) to (e) to (i) to (j) to (o) to (n) to (g) to (f) to (m) to (h) to (3, 5);
				
				\path (3,6) node[circle, fill = none, draw=none, inner sep= .7mm, rotate=90] (dots) {...};
				
				\path (23,1) node[circle, fill = white, draw=black, inner sep= .7mm] (b) {};
				\path (23,-1) node[circle, fill = white, draw=black, inner sep= .7mm] (c) {};
				\path (19,1) node[circle, fill = white, draw=black, inner sep= .7mm] (d) {};
				\path (19,-1) node[circle, fill = white, draw=black, inner sep= .7mm] (e) {};
				\path (17,1) node[circle, fill = white, draw=black, inner sep= .7mm] (f) {};
				\path (17,-1) node[circle, fill = white, draw=black, inner sep= .7mm] (g) {};
				\path (19,3) node[circle, fill = white, draw=black, inner sep= .7mm] (h) {};
				\path (19,-3) node[circle, fill = white, draw=black, inner sep= .7mm] (i) {};
				\path (19,-5) node[circle, fill = white, draw=black, inner sep= .7mm] (j) {};
				\path (17,3) node[circle, fill = white, draw=black, inner sep= .7mm] (m) {};
				\path (17,-3) node[circle, fill = white, draw=black, inner sep= .7mm] (n) {};
				\path (17,-5) node[circle, fill = white, draw=black, inner sep= .7mm] (o) {};
				
				\path (24,0) node[circle, fill = none, draw=none, inner sep= .7mm] (dots) {...};
				
				\draw[] (24, -1) to (c)  to (b) to (d) to (e) to (i) to (j) to (o) to (n) to (g) to (f) to (m) to (h) to (19, 5);
				
				\path (19,6) node[circle, fill = none, draw=none, inner sep= .7mm, rotate=90] (dots) {...};
				
				\draw[dashed] (2, -6)  to (2, 4) to (20, 4) to (20,-6) to (2, -6);
				
			\end{tikzpicture}
			\caption{}
		\end{center}
	\end{subfigure}
	\begin{subfigure}[b]{0.49\textwidth}
		\begin{center}
			\begin{tikzpicture}
				[scale=.25, inner sep=1mm,
				cherry/.style={circle,draw=black,fill=red,thick},
				blueberry/.style={circle,draw=black,fill=blue,thick},
				every loop/.style={min distance=8mm,in=60,out=120,looseness=10}]
				
				\path (-1,1) node[circle, fill = white, draw=black, inner sep= .7mm] (b) {};
				\path (-1,-1) node[circle, fill = white, draw=black, inner sep= .7mm] (c) {};
				\path (3,1) node[circle, fill = white, draw=black, inner sep= .7mm] (d) {};
				\path (3,-1) node[circle, fill = white, draw=black, inner sep= .7mm] (e) {};
				\path (5,1) node[circle, fill = white, draw=black, inner sep= .7mm] (f) {};
				\path (5,-1) node[circle, fill = white, draw=black, inner sep= .7mm] (g) {};
				\path (3,3) node[circle, fill = white, draw=black, inner sep= .7mm] (h) {};
				\path (3,-3) node[circle, fill = white, draw=black, inner sep= .7mm] (i) {};
				\path (3,-5) node[circle, fill = white, draw=black, inner sep= .7mm] (j) {};
				\path (5,3) node[circle, fill = white, draw=black, inner sep= .7mm] (m) {};
				\path (5,-3) node[circle, fill = white, draw=black, inner sep= .7mm] (n) {};
				\path (5,-5) node[circle, fill = white, draw=black, inner sep= .7mm] (o) {};
				
				\path (-2,0) node[circle, fill = none, draw=none, inner sep= .7mm] (dots) {...};
				
				\draw[] (-2, 1) to (b)  to (c) to (e) to (i) to (j) to (o) to (n) to (g) to (f) to (d) to (h) to (m) to (5, 5);
				
				\path (5,6) node[circle, fill = none, draw=none, inner sep= .7mm, rotate=90] (dots) {...};
				
				\path (23,1) node[circle, fill = white, draw=black, inner sep= .7mm] (b) {};
				\path (23,-1) node[circle, fill = white, draw=black, inner sep= .7mm] (c) {};
				\path (19,1) node[circle, fill = white, draw=black, inner sep= .7mm] (d) {};
				\path (19,-1) node[circle, fill = white, draw=black, inner sep= .7mm] (e) {};
				\path (17,1) node[circle, fill = white, draw=black, inner sep= .7mm] (f) {};
				\path (17,-1) node[circle, fill = white, draw=black, inner sep= .7mm] (g) {};
				\path (19,3) node[circle, fill = white, draw=black, inner sep= .7mm] (h) {};
				\path (19,-3) node[circle, fill = white, draw=black, inner sep= .7mm] (i) {};
				\path (19,-5) node[circle, fill = white, draw=black, inner sep= .7mm] (j) {};
				\path (17,3) node[circle, fill = white, draw=black, inner sep= .7mm] (m) {};
				\path (17,-3) node[circle, fill = white, draw=black, inner sep= .7mm] (n) {};
				\path (17,-5) node[circle, fill = white, draw=black, inner sep= .7mm] (o) {};
				
				\path (24,0) node[circle, fill = none, draw=none, inner sep= .7mm] (dots) {...};
				
				\draw[] (24, 1) to (b)  to (c) to (e) to (i) to (j) to (o) to (n) to (g) to (f) to (d) to (h) to (m) to (17, 5);
				
				\path (17,6) node[circle, fill = none, draw=none, inner sep= .7mm, rotate=90] (dots) {...};
				
				\draw[dashed] (2, -6)  to (2, 4) to (20, 4) to (20,-6) to (2, -6);
				
			\end{tikzpicture}
			\caption{}
		\end{center}
	\end{subfigure}
	\caption{Structure $B$ (see $(a)$). Visiting structure $B$ in any mode allows it to optimally visit an $H$ structure above or below (see $(b)$ and $(c)$). 
		\label{fig:B}}
\end{figure}

\begin{figure}[]
		\begin{center}
			\begin{tikzpicture}[scale=0.45]
				\path (-.5, 0.5) node[circle, fill = white, draw=black, inner sep=.5mm] (a) {};
				\path (3.5, 0.5) node[circle, fill = white, draw=black, inner sep=.5mm] (b) {};
				\draw (a) to (b);
				\path (-.5, -14.5) node[circle, fill = white, draw=black, inner sep=.5mm] (a) {};
				\path (3.5, -14.5) node[circle, fill = white, draw=black, inner sep=.5mm] (b) {};
				\draw (a) to (b);
				\path (-1.5, .5) node[circle, fill = white, draw=none, inner sep=.5mm] (a) {$X$};
				\path (-1.5, -14.5) node[circle, fill = white, draw=none, inner sep=.5mm] (a) {$Y$};
				
				\path (20, 0) node[circle, fill = white, draw=none, inner sep=0mm] (right) {};
				\path (20, 1) node[circle, fill = white, draw=none, inner sep=0mm] (bisright) {};
				\path (20, -4) node[circle, fill = white, draw=none, inner sep=0mm] (right2) {};
				\path (20, -3) node[circle, fill = white, draw=none, inner sep=0mm] (bisright2) {};
				\path (20, -8) node[circle, fill = white, draw=none, inner sep=0mm] (right3) {};
				\path (20, -7) node[circle, fill = white, draw=none, inner sep=0mm] (bisright3) {};
				\path (20, -9) node[circle, fill = white, draw=none, inner sep=0mm] (right4) {};
				\path (20, -8) node[circle, fill = white, draw=none, inner sep=0mm] (bisright4) {};
				\path (20, -12) node[circle, fill = white, draw=none, inner sep=0mm] (right5) {};
				\path (20, -11) node[circle, fill = white, draw=none, inner sep=0mm] (bisright5) {};
				\path (20, -15) node[circle, fill = white, draw=none, inner sep=0mm] (right6) {};
				\path (20, -14) node[circle, fill = white, draw=none, inner sep=0mm] (bisright6) {};
				
				\path (4, 0) node[circle, fill = white, draw=none, inner sep=0mm] (left) {};
				\path (4, 1) node[circle, fill = white, draw=none, inner sep=0mm] (bisleft) {};
				\path (4, -4) node[circle, fill = none, draw=none, inner sep=0mm] (left2) {};
				\path (4, -3) node[circle, fill = none, draw=none, inner sep=0mm] (bisleft2) {};
				\path (4, -8) node[circle, fill = white, draw=none, inner sep=0mm] (left3) {};
				\path (4, -7) node[circle, fill = white, draw=none, inner sep=0mm] (bisleft3) {};
				\path (4, -9) node[circle, fill = white, draw=none, inner sep=0mm] (left4) {};
				\path (4, -8) node[circle, fill = white, draw=none, inner sep=0mm] (bisleft4) {};
				\path (4, -12) node[circle, fill = white, draw=none, inner sep=.5mm] (left5) {};
				\path (4, -11) node[circle, fill = white, draw=none, inner sep=.5mm] (bisleft5) {};
				\path (4, -15) node[circle, fill = white, draw=none, inner sep=0mm] (left6) {};
				\path (4, -14) node[circle, fill = white, draw=none, inner sep=0mm] (bisleft6) {};
				
				\path (5, -11) node[circle, fill = white, draw=none, inner sep=.5mm, rotate=90] (dots) {\footnotesize ...};
				\path (10, -11) node[circle, fill = white, draw=none, inner sep=.5mm, rotate=90] (dots) {\footnotesize ...};
				\path (19, -11) node[circle, fill = white, draw=none, inner sep=.5mm, rotate=90] (dots) {\footnotesize ...};
				
				\path (12, .5) node[circle, fill = white, draw=none, inner sep=.5mm] (dots) {\footnotesize $C_1$};
				\path (12, -3.5) node[circle, fill = white, draw=none, inner sep=.5mm] (dots) {\footnotesize $C_2$};
				\path (12, -7.5) node[circle, fill = white, draw=none, inner sep=.5mm] (dots) {\footnotesize $C_3$};
				\path (12, -14.5) node[circle, fill = white, draw=none, inner sep=.5mm] (dots) {\footnotesize $C_n$};

				//drawing lines
				\draw[] (left)  to (right);
				\draw[] (bisleft)  to (bisright);
				\draw[] (left2)  to (right2);
				\draw[] (bisleft2)  to (bisright2);
				\draw[] (left3)  to (right3);
				\draw[] (bisleft3)  to (bisright3);
				\draw[] (right6)  to (left6);
				\draw[] (bisleft6)  to (bisright6);
				
				//drawing 1chains connecting Ci
				
				\path (20.5, .5) node[circle, fill = white, draw=black, inner sep=.5mm, rotate=90] (a) {};
				\path (20.5, -3.5) node[circle, fill = white, draw=black, inner sep=.5mm, rotate=90] (b) {};
				\draw[] (a)  to (22.5, .5) to (22.5, -3.5) to (b);
				
				\path (3.5, -3.5) node[circle, fill = white, draw=black, inner sep=.5mm, rotate=90] (a) {};
				\path (3.5, -7.5) node[circle, fill = white, draw=black, inner sep=.5mm, rotate=90] (b) {};
				\draw[] (a)  to (1.5, -3.5) to (1.5, -7.5) to (b);
				
				\path (22.5, -11) node[circle, fill = white, draw=none, inner sep=.5mm, rotate=90] (dots) {\footnotesize ...};
				\path (20.5, -7.5) node[circle, fill = white, draw=black, inner sep=.5mm, rotate=90] (a) {};
				\path (20.5, -14.5) node[circle, fill = white, draw=black, inner sep=.5mm, rotate=90] (b) {};
				\draw[] (a)  to (22.5, -7.5) to (22.5, -10);
				\draw[] (22.5,-12)  to (22.5, -14.5) to (b);
				
				//arrow with distances a and b
				
				\draw[<->] (-.5,1)  to (3.5, 1);
				\path (1.5, 1.5) node[fill = white, draw=none, inner sep=.5mm] (dots) {\footnotesize $1000$};
				\draw[<->] (1.5, -3)  to (3.5, -3);
				\path (2.5, -2.5) node[fill = white, draw=none, inner sep=.5mm] (dots) {\footnotesize $100$};
				\draw[<->] (6.5, -1.5)  to (8.5, -1.5);
				\path (7.5, -1) node[fill = white, draw=none, inner sep=.5mm] (dots) {\footnotesize $100$};
				\draw[<->] (5, -1)  to (5, -0.25);
				\path (5.8, -.6) node[fill = white, draw=none, inner sep=.5mm] (dots) {\footnotesize $22$};
				\draw[<->] (5, -2)  to (5, -2.75);
				\path (5.8, -2.5) node[fill = white, draw=none, inner sep=.5mm] (dots) {\footnotesize $21$};
				
				// H structures
				
				\path (4.9, -1.5) node[rectangle, fill = white, draw=black, inner sep=.5mm] (dots) {\footnotesize $H_{1,1}$};
				\path (10, -1.5) node[rectangle, fill = white, draw=black, inner sep=.5mm] (dots) {\footnotesize $H_{1,2}$};
				\path (19, -1.5) node[rectangle, fill = white, draw=black, inner sep=.5mm] (dots) {\footnotesize $H_{1,m}$};
				\path (14, -1.5) node[rectangle, fill = white, draw=none, inner sep=.5mm] (dots) {\footnotesize $...$};
				
				\path (4.9, -5.5) node[rectangle, fill = white, draw=black, inner sep=.5mm] (dots) {\footnotesize $H_{2,1}$};
				\path (10, -5.5) node[rectangle, fill = white, draw=black, inner sep=.5mm] (dots) {\footnotesize $H_{2,2}$};
				\path (19, -5.5) node[rectangle, fill = white, draw=black, inner sep=.5mm] (dots) {\footnotesize $H_{2,m}$};
				\path (14, -5.5) node[rectangle, fill = white, draw=none, inner sep=.5mm] (dots) {\footnotesize $...$};
				
				\path (4.9, -9.5) node[rectangle, fill = white, draw=black, inner sep=.5mm] (dots) {\footnotesize $H_{3,1}$};
				\path (10, -9.5) node[rectangle, fill = white, draw=black, inner sep=.5mm] (dots) {\footnotesize $H_{3,2}$};
				\path (19, -9.5) node[rectangle, fill = white, draw=black, inner sep=.5mm] (dots) {\footnotesize $H_{3,m}$};
				\path (14, -9.5) node[rectangle, fill = white, draw=none, inner sep=.5mm] (dots) {\footnotesize $...$};
				
				\path (5.5, -12.5) node[rectangle, fill = white, draw=black, inner sep=.5mm] (dots) {\footnotesize $H_{n-1,1}$};
				\path (10.5, -12.5) node[rectangle, fill = white, draw=black, inner sep=.5mm] (dots) {\footnotesize $H_{n-1,2}$};
				\path (18.2, -12.5) node[rectangle, fill = white, draw=black, inner sep=.5mm] (dots) {\footnotesize $H_{n-1,m}$};
				\path (14, -12.5) node[rectangle, fill = white, draw=none, inner sep=.5mm] (dots) {\footnotesize $...$};
			\end{tikzpicture}
		\end{center}
		\caption{\ConstructionOne from \ExactCover to \ContinuousVectorTSPWithStops (before adding structures $A$ and $B$ and rotating the city set).}
		\label{fig:VTSP_construction}
\end{figure}

Construct the structure shown in Figure~\ref{fig:VTSP_construction}, where $n$ is the number of subsets given in the corresponding \ExactCover instance, and $m$ the number of elements in the universe.

Then, for every 2-chain $C_i$, remove the cities positioned directly above or below an $H$. In other words, if the $H$ structure is contained in a bounding box of coordinates $(x_1, y_1)$ to $(x_2, y_2)$, the cities of $C_i$ between $x$-coordinates $x_1$ and $x_2$ are concerned. Insert in this emptied space one of two structures, depending on the elements in $F_i$, $C_i$'s corresponding subset. If the subset contains the element corresponding to the above (or below) $H$, then replace by structure $A$ (see Figure~\ref{fig:A}), otherwise replace by structure $B$ (see Figure~\ref{fig:B}). 
Finally, rotate the whole city set by 45°.
As the city coordinates should be integer, after the rotation we scale everything by $\sqrt{2}$. (The rotation and the scaling are the only changes we add \textit{w.r.t.} Papadimitriou's original construction.)
This concludes \ConstructionOne.\\

The following definitions are from Papadimitriou~\cite{papadimitriou1977euclidean}. A subset $Q$ of the set of cities is an \textit{$a$-component} if for all $q \in Q$ we have $\texttt{min}(\texttt{cost}(q, p) : p \not\in Q) \geq a$ and $\texttt{max}(\texttt{cost}(q, p) : p \in Q) < a$.
A \textit{$k$-path} for a set of cities is a set of $k$, not necessarily closed paths visiting all cities. A solution for a \ContinuousVectorTSPWithStops instance is thus a closed (or cyclic) 1-path.
A subset of cities is \textit{$a$-compact} if, for all positive integers $k$, an optimal $k$-path has cost less than the cost of an optimal $(k+1)$-path plus $a$. Note that $a$-components are trivially $a$-compact.

\begin{lemma}[Papadimitriou~\cite{papadimitriou1977euclidean}]
	Suppose we have $N$ $a$-components $P_1, ..., P_N \subseteq P$, such that the cost to connect any two components is at least $2a$, and $P_0 = P \setminus \bigcup\limits_{i=1}^{N} P_i$
	is $a$-compact. Suppose that any optimal 1-path of this instance does not directly connect any two $a$-components (but rather indirectly connects them through $P_0$). Let $K_1, ..., K_N$ be the costs of the optimal 1-paths of $P_1, ..., P_N$ and $K_0$ the cost of the optimal $(N+ 1)$-path of $P_0$. If there is a 1-path $S$ of $P$ consisting of the union of an optimal $(N+1)$-path of $P_0$, $N$ optimal 1-paths of $P_1, ..., P_N$ and $2N$ paths of cost $a$ connecting $a$-components to $P_0$, then $S$ is optimal. If no such 1-path exists, the optimal 1-path of $P$ has a cost greater than $K = K_0 + K_1 + ... + K_N + 2Na$.
	\label{lem:papa_lemma}
\end{lemma}

Papadimitriou proposed the following optimal $k$-paths in his construction: 
\begin{itemize}
	\item For 1-chains, the 1-path shown in \Cref{fig:1-chain};
	\item For 2-chains, the path zigzagging one way (mode 1) or the other (mode 2) as shown in \Cref{fig:2-chain};
	\item For $A$ structures, a 1-path for mode 1, and a 2-path for mode 2, such as those shown in \Cref{fig:A}. Note that no 2-path exists for mode 1 such that the endpoints would align with vertices of a structure $H$ directly above or below (as they do for the 2-path in mode 2); 
	\item For $B$ structures, independent of the mode, a 2-path aligning with vertices of a structure $H$ above or below such as shown in \Cref{fig:B};
	\item For $H$ structures, the 1-path shown in \Cref{fig:H}. Among all 1-paths for $H$ having as endpoints two of the cities $a, a', b, b', c, c', d, d'$, there are 4 optimal 1-paths, namely those with endpoints $(a, a')$, $(b, b')$, $(c, c')$, $(d, d')$.
\end{itemize} 
It can be shown that these $k$-paths are optimal for \ContinuousVectorTSPWithStops as well, although due to the nature of \ContinuousVectorTSPWithStops, other paths would have been optimal as well before rotating the city set, for example for 2-chains, the path shown in \Cref{fig:2-chain_undesiredpaths} is optimal, but it completely disobeys any ''mode``. After rotation, only the described $k$-paths are optimal (see again \Cref{fig:2-chain_undesiredpaths}). 

\begin{figure}[h]
	\begin{center}
		\begin{subfigure}[b]{0.49\textwidth}
			\begin{center}
				\begin{tikzpicture}
					[scale=.45, inner sep=1mm,
					cherry/.style={circle,draw=black,fill=red,thick},
					blueberry/.style={circle,draw=black,fill=blue,thick},
					every loop/.style={min distance=8mm,in=60,out=120,looseness=10}]
					
					\path (-.5,.5) node[circle, fill = white, draw=black, inner sep= .7mm] (a) {};
					\path (0,1) node[circle, fill = white, draw=black, inner sep= .7mm] (b) {};
					\path (0,0) node[circle, fill = white, draw=black, inner sep= .7mm] (c) {};
					\path (2,1) node[circle, fill = white, draw=black, inner sep= .7mm] (d) {};
					\path (2,0) node[circle, fill = white, draw=black, inner sep= .7mm] (e) {};
					\path (4,1) node[circle, fill = white, draw=black, inner sep= .7mm] (f) {};
					\path (4,0) node[circle, fill = white, draw=black, inner sep= .7mm] (g) {};
					
					\path (6,0.5) node[circle, fill = none, draw=none, inner sep= .7mm] (A) {...};
					
					\draw (a) to (b);
					\draw (b) to (c);
					\draw (c) to (e);
					\draw (e) to (d);
					\draw (d) to (g);
					\draw (f) to (g);
					\draw (f) to (6,1);
				\end{tikzpicture}
				\caption{}
			\end{center}
		\end{subfigure}
		\begin{subfigure}[b]{0.49\textwidth}
			\begin{center}
				\begin{tikzpicture}
					[rotate=45, scale=.45, inner sep=1mm,
					cherry/.style={circle,draw=black,fill=red,thick},
					blueberry/.style={circle,draw=black,fill=blue,thick},
					every loop/.style={min distance=8mm,in=60,out=120,looseness=10}]
					
					\path (-.5,.5) node[circle, fill = white, draw=black, inner sep= .7mm] (a) {};
					\path (0,1) node[circle, fill = white, draw=black, inner sep= .7mm] (b) {};
					\path (0,0) node[circle, fill = white, draw=black, inner sep= .7mm] (c) {};
					\path (2,1) node[circle, fill = white, draw=black, inner sep= .7mm] (d) {};
					\path (2,0) node[circle, fill = white, draw=black, inner sep= .7mm] (e) {};
					\path (4,1) node[circle, fill = white, draw=black, inner sep= .7mm] (f) {};
					\path (4,0) node[circle, fill = white, draw=black, inner sep= .7mm] (g) {};
					
					\path (6,0.5) node[circle, fill = none, draw=none, inner sep= .7mm] (A) {...};
					
					\draw (a) to (c);
					\draw (b) to (c);
					\draw (b) to (d);
					\draw (e) to (d);
					\draw (e) to (g);
					\draw (f) to (g);
					\draw (f) to (6,1);
				\end{tikzpicture}
				\caption{}
			\end{center}
		\end{subfigure}
	\end{center}
	\caption{Example of an undesired but optimal 1-path for a 2-chain before rotation $(a)$. Papadimitrou's proposed optimal 1-path (in mode 2) after rotation $(b)$. \label{fig:2-chain_undesiredpaths}}
\end{figure}

These optimal $k$-paths in our \ContinuousVectorTSPWithStops construction have a cost of:
\begin{itemize}
	\item $2(n-1)$ for a 1-chain of $n$ cities, or $2\ell$ for a 1-chain of length $\ell$. Let's denote the latter as $c_{1, \ell}$;
	\item $6 + (2 + 2 \sqrt{2})\frac{n-4}{2}$ for a 2-chain of $n$ cities, or $6 + (2 + 2 \sqrt{2})\frac{\ell-1}{2}$ for a 2-chain of length $\ell$. Let's denote the latter as $c_{2, \ell}$;
	\item $52 + 2\sqrt{2}$ for a structure $A$ in mode 1 (\textit{resp.} $52$ in mode 2), or $42 - 6\sqrt{2}$ (\textit{resp.} $42 - 8{\sqrt{2}}$) more than the part of the 2-chain tour it replaces. Let's denote the latter as $c_{A, 1}$ (\textit{resp.} $c_{A, 2}$);
	\item $36$ for a structure $B$ independent of the mode, or $26 - 8{\sqrt{2}}$ more than the part of the 2-chain tour it replaces. Let's denote the latter as $c_B$;
	\item $36 + 8 \sqrt{{2}} + 2\sqrt{6}$ for a structure $H$. Let's denote this cost as $c_H$.
\end{itemize}

We are now ready to prove the main theorem based on the presented construction and lemmas.

\begin{theorem}
	\ContinuousVectorTSPWithStops is NP-hard.
	\label{theorem:CVTSP0_nphard}
\end{theorem}

\begin{proof}
	We reduce from \ExactCover with \ConstructionOne.
	We show that this construction fits the hypothesis of \Cref{lem:papa_lemma}. 
	Observe that connecting a $k$-path from some 2-chain $C_i$ (with structures $A$ and/or $B$) with a $k$-path from some $H_{i,j}$ (or $H_{i-1,j}$) has optimal cost $4\sqrt{5}$. This optimal cost is only attained by connecting some extremity city of structure $H$ ($a, a', b, b', c, c', d, d'$ in \Cref{fig:H}) to the closest city of 2-chain $C$ positioned at distance 20 with a 45° angle. 
	The optimal cost to connect two $k$-paths between two $H_{i,j}$, is more than $4\sqrt{5}$ (it is at least $4\sqrt{11}$ to connect $H_{i,j}$ and $H_{i+1,j}$ and at least $20$ to connect $H_{i,j}$ and $H_{i,j+1}$) vectors.
	It should be clear that an optimal 1-path must have $X$ and $Y$ as endpoints (see again \Cref{fig:VTSP_construction}). 
	This construction meets the hypotheses of Lemma~\ref{lem:papa_lemma} with $a = 4\sqrt{5}$, $N = m(n-1)$, $K_1 = ... = K_N = c_H$ and $K_0 = 2c_{1, 1000} + 2(n-1)c_{1, 100} + (n-1)c_{1, 51} + nc_{2, 8m+100(m-1)} + mc_{A, 1} + (p-m)c_{A,2} + (nm-p)c_B$, where $p$ is the sum of cardinalities of all given subsets of the \ExactCover instance, which in turn describes the ratio of $A$ structures vs $B$ structures, of which the costs are distinct. Note that some distances in the construction (such as 22 and 21) were chosen very carefully (so the cost between $A$ (or $B$) structures and $H$ structures is exactly $a$, as the former are larger by 2 units at the bottom, and 1 unit at the top \textit{w.r.t.} 2-chains).
	We examine when this structure has an optimal 1-path $S$, as described in the lemma. $S$ traverses all 1-chains in the obvious way, and all 2-chains in one of the two modes. Since its portion on $P_0$ has to be optimal, $S$ must visit a structure $H$ from any structure $B$ encountered, and it must return to the symmetric city of $B$, since its portion on $H$ must be optimal, too. If $S$ encounters a structure $A$ and the corresponding chain is traversed in mode 2, $S$ will also visit a structure $H$. However, if the corresponding 2-chain is traversed in mode 1, $S$ will traverse $A$ without visiting any structure $H$, since trajectories connecting $P_0$ and any $H$ structure must be of cost $a$. Moreover this must happen exactly once for each column of the structure, since there are $n-1$ copies of $H$ and $n$ structures $A$ or $B$ in each column. (In particular, this implies that an $H$ structure not visited from the current 2-chain must be visited from the following 2-chain.) Hence, if we consider the fact that $C_j$ is traversed in mode $1$ (\textit{resp.} mode $2$) to mean that the corresponding subset is (\textit{resp.} is not) contained in the \ExactCover solution, we see that the existence of a 1-path $S$, as described in \cref{lem:papa_lemma}, implies the \ExactCover instance admits a solution. 
	Conversely, if the \ExactCover instance admits a solution, we assign, as above, modes to the chains according to whether or not the corresponding subset is included in the solution. It is then possible to exhibit a 1-path $S$ meeting the requirements of \cref{lem:papa_lemma}. Hence the structure at hand has a 1-path of cost no more than $K = K_0 + m(n-1)(2a + c_H)$ if and only if the given instance of \ExactCover is positive.
	Finally, to obtain a valid \ContinuousVectorTSPWithStops tour, connect both endpoints $X$ and $Y$ in \cref{fig:VTSP_construction} with a 1-chain, and increase $K$ accordingly.
\end{proof}

\begin{lemma}
	\label{lemma:CVTSP0_gap}
	The difference in cost between positive \ContinuousVectorTSPWithStops instances and negative ones (of a same size) in \ConstructionOne is at least $\sqrt{42} - 2\sqrt{10} > 0.15$.
\end{lemma}

\begin{proof}
	In Papadimitriou's construction, it is stated that the minimum precision necessary is $\sqrt{401} - 20 > 0.02$, which corresponds to distinguishing between optimally connecting a 1-path from a structure $A$ or $B$ to a 1-path from a structure $H$, or connecting it suboptimally (all other suboptimal connections/paths can be recognized with lower precision). This implies that a gap between positive and negative instances must be at least this large.\\
	In \ConstructionOne, this distinction remains the smallest one (\textit{i.e.} the one requiring the highest precision), and the gap evaluates to $\sqrt{42} - 2\sqrt{10} > 0.15$. 
\end{proof}

We remark that scaling up the construction by some factor $r^2$ multiplies the cost of a solution tour, as well as the gap, by $r$.

\begin{lemma}
	\label{lemma:optimalCVTSP0K_optimalVTSP0Kplusn}
	If an optimal \ContinuousVectorTSPWithStops tour has cost $K$ on an instance of $n$ cities, then an optimal \VectorTSPWS trajectory has a cost of at most $K + n$.
\end{lemma}

\begin{proof}
	The cost between any pair of cities $p_i$, $p_j$ in \VectorTSPWS is $c(p_i, p_j) = \lceil 2\sqrt{K} \rceil$ where $K$ is the largest distance between the cities among all dimensions, and for \ContinuousVectorTSPWithStops this cost is $2\sqrt{K}$. Note that the difference between such costs is at most 1. An optimal tour in \ContinuousVectorTSPWithStops is exactly $n$ intercity costs (in some order), meaning the difference between optimal trajectories is at most $n$.
\end{proof}

\begin{theorem}
	\label{theorem:VTSP0_nphard}
	\VectorTSPWS is NP-hard.
\end{theorem}

\begin{proof}
	Take \ConstructionOne and scale it up by a sufficiently large factor, say $R = (100n)^2$. The gap from \Cref{lemma:CVTSP0_gap} is now at least $13n$. 
	Since \ContinuousVectorTSPWithStops is NP-hard (\Cref{theorem:CVTSP0_nphard}), and since the cost of an optimal \VectorTSPWS trajectory is at most that of an optimal \ContinuousVectorTSPWithStops tour plus $n$ (\Cref{lemma:optimalCVTSP0K_optimalVTSP0Kplusn}), which fits within the gap of size at least $13n$, \VectorTSPWS is also NP-hard.
\end{proof}

Note that a gap of any (constant) size is attainable for \VectorTSPWS, by changing $R$ in the proof of \Cref{theorem:VTSP0_nphard} accordingly. This will be useful for proving that \VectorTSP is NP-hard in \Cref{subsec:NPhard_VTSPinf}.

\subsection{NP-hardness of \VectorTSP}
\label{subsec:NPhard_VTSPinf}

Through a scaling argument and use of the gap from \Cref{subsec:NPhard_VTSP0}, we show that the NP-hardness from \VectorTSPWS (\Cref{theorem:VTSP0_nphard}) transfers to \VectorTSP.

\textbf{\ConstructionTwo} and \textbf{\ConstructionThree:}
Consider \ConstructionOne. Define grid factor $G = 100n^2$, and the scaling factor $S = \frac{100}{0.13}G^2 n^3$. As established in \Cref{subsec:NPhard_VTSP0}, the gap between positive \VectorTSPWS instances and negatives ones can be increased by scaling up the construction sufficiently; scale until the gap is at least $(nG)^2$.
This concludes \ConstructionTwo.
Scale up again by $S^2$ and replace each city with a $G\times G$ square of cities (so containing $(G+1)^2$ cities) with the left top corner city in the original city location.
This concludes \ConstructionThree.

We define each square of size $5G \times 5G$ centred on such a square of cities as a \textbf{region}, the square itself as the \textbf{inner region} and the region without the inner region as the (empty) \textbf{outer region}. (See also \Cref{fig:region_visits}.)
A \textbf{region visit} is a segment of trajectory intersecting the inner region; it is counted from the entry into the corresponding
region until the trajectory  leaves the region (a single region visit can leave and re-enter the inner region).
A region visit is \textbf{fast} if a speed component (in some dimension) is always larger than zero during the entire region visit. Otherwise the region visit is called \textbf{slow}. Note that a slow visit might never come to a complete stop, as speed components may turn to zero at different points in time.

\begin{figure}[h]
	\begin{center}
		\begin{subfigure}[b]{0.49\textwidth}
			\begin{center}
				\includegraphics[width=\textwidth]{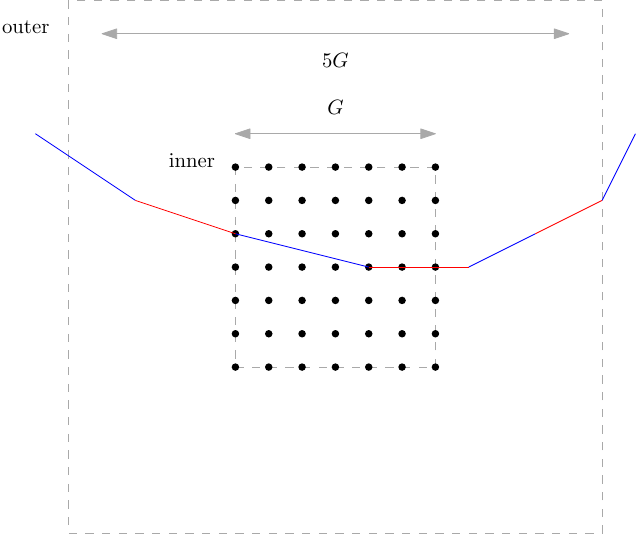}
				\caption{Fast visit: at least one speed component (the one along the $x$ coordinate in this example) remains non-zero.}
			\end{center}
		\end{subfigure}
		\begin{subfigure}[b]{0.49\textwidth}
			\begin{center}
				\includegraphics[width=\textwidth]{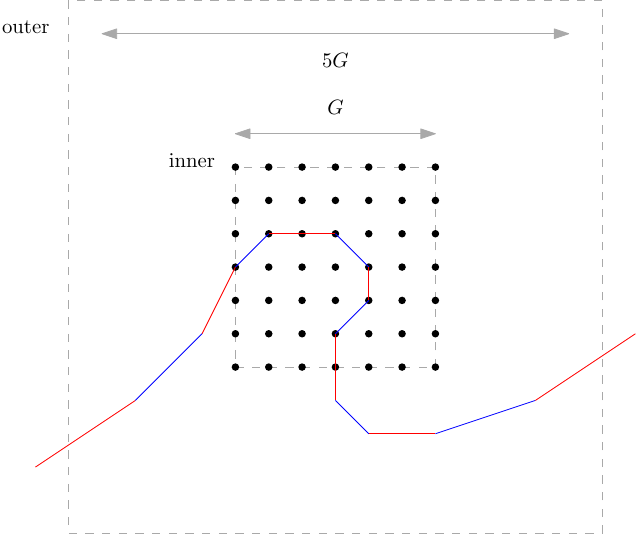}
				\caption{Slow visit: both speed components are zero at some time during the visit (although not at the same time in this example).}
			\end{center}
		\end{subfigure}
	\end{center}
	\caption{Illustrative examples of a region and visits of a region. The larger grey dashed square indicates the outer region, while the other square shows the inner region containing the cities. Note that due to scaling, the outer region is much larger in actuality than depicted. The alternating red and blue lines are successive segments of trajectories.
	\label{fig:region_visits}}
\end{figure}

\begin{lemma}
	Any slow region visit can be locally converted to a region visit that stops within the inner region, then visits all the
	cities, stopping at each one. The conversion doesn’t change the trajectory outside the region visit	and costs at most $2G^2 + 15G$ additional vectors.
\end{lemma}

\begin{proof}
	A slow region visit means that all speed components must be zero at some point in time (although not necessarily at the same time). Take the first time it comes to a stop in some dimension. Keep this dimension to a stop (instead of whatever acceleration it had during the slow visit) until the other dimension comes to a stop (as it does in the original slow visit). Now the vehicle is at a stop somewhere in the region. Using at most $5G$ vectors, the vehicle can travel to the inner region and come at a stop there at a corner city. Visiting the cities then requires at most $2G^2$ vectors. Finally, we describe how to make the vehicle join back up with the original slow region visit at the latest configuration where a speed component is zero, say \textit{w.l.o.g.} configuration $c = (x, y, dx, 0)$. Move the vehicle to $(x-\frac{dx (|dx|+ 1)}{2}, y, 0, 0)$ using at most $5G$ vectors since the position is in the region as in the original slow visit, the vehicle was able to accelerate from zero speed to $dx$ speed at position $(x,y)$ while staying inside the region. From that configuration it can reach $c$ using at most $5G$ vectors, linking it back up to the original slow visit. 
\end{proof}

\begin{lemma}
	A fast region visit can visit at most $G$ cities.
\end{lemma}

\begin{proof}
	\textit{W.l.o.g.} suppose the speed component which is always larger than zero to be $dx$.
	The position coordinate $x$ is thus monotonically increasing or monotonically decreasing during the fast region visit, meaning it can visit at most one city per coordinate on the $x$ axis, of which there are at most $G$ in a region.
\end{proof}

\begin{corollary}
	If a trajectory visits all cities in a region, it must do so either with at least one slow region visit, or with at least $G$ fast region visits.
\end{corollary}

\begin{lemma}
	A segment of a trajectory containing five fast region visits to the same region of cost $K$, can be replaced by a slow region visit of cost at most $K+13G$.
\end{lemma}

\begin{proof}
	If a trajectory leaves the region and then returns, the velocity along at least one dimension has changed sign (passing through the zero in the process). 
	Moreover, if the sign of the velocity along some other dimension has not changed, then this speed must be below $G$, as leaving and returning to the inner region takes more than one unit of time.
	Therefore between each two region visits for each dimension there is a moment with a relatively low speed along
	this dimension (not necessarily the same moment).
	
	For each dimension we independently modify the trajectory as follows. After the low speed moment between
	the first two visits, try to get within the inner region with zero speed (both — along this coordinate) as fast as possible.
	This cannot take longer than until the second low-speed moment plus time $2G$ in case there was no full stop
	in between. 
	We do the same transformation in reversed
	time from the fourth low-speed moment towards the third one.
	Note that without changing the trajectory outside of the segment, these two modifications will replace at most $6G$ vectors into
	the third visit. (This in principle can lead to an overlap.) Now we shift the second replaced segment and
	everything afterwards by $6G$, by simply adding null vectors in-between the first replaced segment and the second one.
	
	Observe that thanks to this idle time in all dimensions, we must obtain a point in time where the modified trajectory is within the inner region with zero speed (in all dimensions), which is thus a slow region visit.
\end{proof}

Note that the new slow region visit may not visit other regions/cities potentially visited by the original segment.

\begin{lemma}
	Given a solution of cost $K$ for \VectorTSP in \ConstructionThree, one can construct a solution of cost at most $K + 4nG^2$ for \VectorTSPWS in \ConstructionThree.
\end{lemma}

\begin{proof}
	Consider all the regions and all the region visits. 
	Let’s call a region processed whenever it has been modified so that the cities are visited with stops.
	All the slow region visits can be processed at the cost of $3G^2$ per region. 
	Afterwards, segments with five fast region visits
	to the same unprocessed region, no overlap to any already chosen segments, and the minimum number of visits to
	other unprocessed regions, are modified to a slow region visit at the cost of $G^2$, which in turn can again be processed at the cost of $3G^2$ extra vectors. 
	Note that if some region remains unprocessed, then we have processed fewer than $n$ regions, and each processed region can spend no more than 4 visits to each of the other regions. Thus each unprocessed region has $G - 4n$ usable fast region visits, split
	into at most $n$ groups by the previously chosen slow region visits and segments of five fast region visits. Since $\frac{G-4n}{n} > 5$, we can always continue to process cities until no unprocessed city remains.
\end{proof}

We are now ready to prove the main result.

\begin{theorem}
	\label{theorem:VTSPinf_nphard}
	\VectorTSP is NP-hard.
\end{theorem}

\begin{proof}
	We reduce from \VectorTSPWS, which we have proven NP-hard in \Cref{theorem:VTSP0_nphard}. More precisely, we've proven it has an NP-hard gap, which in \ConstructionTwo we've widened through a scaling argument to be at least $(nG)^2$. It suffices to show that \ConstructionThree has a gap corresponding to this one.
	Consider a solution for \VectorTSPWS in \ConstructionTwo with cost $K$. It can be scaled to a solution for \VectorTSP in \ConstructionThree with cost at most $KS + n + 3nG^2$.
	Consider a solution for \VectorTSP in \ConstructionThree with cost at most $KS + 4nG^2$. We can modify it so it visits cities at a stop for a total cost of at most $KS + 8nG^2$, by forgetting the points except the original city in the corner. Then, scale it down to a solution of cost at most $K + \frac{8nG^2}{S} + nG^2$. This is a solution for \VectorTSPWS in \ConstructionTwo and observe that $\frac{8nG^2}{S} + nG^2 < (nG)^2$.
	Thus we have constructed a gap for \VectorTSP in \ConstructionThree to which the gap for \VectorTSPWS in \ConstructionTwo reduces.
\end{proof}

\section{Algorithmic aspects and experimental investigation}
\label{sec:algorithms}

In this section, we investigate some algorithmic approaches for tackling \VectorTSP. The proposed framework is based on an interaction between a high-level part that compute the tour (i.e. visit order), and a trajectory algorithm that evaluates the cost of \texttt{racetrack($\pi$)}, i.e., the cost of an optimal trajectory realizing this tour. Our solution to the latter problem is perhaps the main contribution of this section, it relies on a multi-point adaptation of the A* algorithm, which may be of independent interest. Experiments are conducted for quantifying how often a trajectory based on an optimal tour for \EuclideanTSP tend to be \emph{suboptimal} for \VectorTSP. The main goal of these experiment is to legitimate the study of \VectorTSP as an independent problem.

\subsection{Computing the tours (\flipTour)}
\label{sec:high-level-algorithm}

A classical heuristic in the TSP literature is the \texttt{Flip} algorithm~\cite{croes1958method} (also called $2$-opt although it is not an approximation algorithm). In each step, every possible \emph{flip}, \textit{i.e.} inversion of a subtour, of the current tour $\pi$ is evaluated. If the resulting tour $\pi'$ improves upon $\pi$, then it is selected and the algorithm tries to improve $\pi'$ in turn. Eventually, the algorithm finds a tour that is a local optimum with respect to the flip operation.
The high level part of our algorithm, called \flipTour, implements a similar approach for exploring the tours for \VectorTSP.  
The main difference between \texttt{Flip} and \flipTour is that the cost of a tour is not evaluated in terms of distance, but in terms of the cost of an optimal trajectory realizing it (obtained by calling the \multipointastar algorithm presented in \Cref{subsec:A-star}). 

Let \texttt{arbitrary\_tour()} be a function that returns a trivial initial tour of length at most $O(nL)$, starting and terminating at the desired city. Such a tour can be obtained easily by using the strategy described in the proof of \Cref{lem:walk}. Let \texttt{racetrack()} be the function that returns the cost of an optimal trajectory for the given tour, and let \texttt{flip()} be the function that inverts the desired segment of the underlying tour. Then \flipTour is described by \cref{algo:2-opt}.

\begin{algorithm}[h]
	\caption{: \flipTour. }
	\label{algo:2-opt}
	Input: a set $P$ of cities and a starting city $p_0 \in P$.\\
	Output: a tour whose optimal trajectory cannot be improved by inverting a subtour.
	\begin{algorithmic}[1]
		\State $\pi_{opt} \gets \texttt{arbitrary\_tour}(P, p_0)$
		\State $C_{opt} \gets \texttt{oracle(}\pi_{opt}\texttt{)}$
		\State $improved \gets \texttt{true}$
		\While{$improved$} 
		\State $improved \gets \texttt{false}$
		\For{each city $i \in P\setminus\{p_0\}$}
		\For{each city $j \in P\setminus\{p_0,i\}$}
		\State $\pi_{test} \gets \texttt{flip(}\pi_{opt}, i, j\texttt{)}$
		\State $C_{test} \gets \texttt{oracle(}\pi_{test}\texttt{)}$
		\If{$C_{test} < C_{opt}$}
		\State $\pi_{opt} \gets \pi_{test}$
		\State $C_{opt} \gets C_{test}$
		\State $improved \gets \texttt{true}$
		\State $\texttt{break}$
		\EndIf\
		\EndFor\
		\If{$improved$}
		\State $\texttt{break}$
		\EndIf\
		\EndFor\
		
		\EndWhile\
		\State \textbf{return} $\pi_{opt}$
	\end{algorithmic}
\end{algorithm}

\begin{theorem}
  The above algorithm always reaches a local optimum after calling the oracle at most $O(n^3L)$ many times,
  where $n$ is the number of cities and $L$ is the maximum distance (in
  a dimension) between two cities.
	\label{theorem:2-opt}
\end{theorem}

\begin{proof}
  The initial tour corresponds to a trajectory of length at most $O(nL)$ (\Cref{lem:walk}), thus $C_{opt}$ is initialized with a value of at most $O(nL)$ in Line~2. Now, the main loop iterates only when a shorter trajectory is found, which can occur at most as many times as the length of the initial trajectory. Finally, in each iteration, at most $O(n^2)$ modified tours are evaluated by the oracle. 
\end{proof}

\subsection{Optimal trajectory for a given tour (\multipointastar)}
\label{subsec:A-star}

Here, we discuss the problem of computing an optimal trajectory that realizes a given tour, i.e., how to compute \texttt{racetrack($\pi$)} for a given $\pi$. 
A previous work of interest is Bekos \textit{et al.}~\cite{bekos2018algorithms}, which addresses the problem of computing an optimal \Racetrack trajectory in the so-called ``Indianapolis'' track, where the track has a fixed width and right-angle turns. The Indianapolis setting makes it possible to decompose the computation by using dynamic programming between the (constant sized) corner areas of the track. In contrast, we consider an open space with no simple way to decompose the computation. Our strategy adapts the \AStar algorithm to searching a multi-point optimal path in the configuration graph. Namely, 
given a sequence of points $\pi=(p_1, p_2, \dots, p_n)$, compute (the cost of) an optimal trajectory visiting the points in order, starting at $p_1$ and ending at $p_n$ at zero speed. (In the special case of \VectorTSP, $p_1$ and $p_n$ coincide.)

Hence, the main difficulty is that the optimal trajectory for a tour in open space does not reduce to gluing together subtrajectories for given parts of the tour. Our contribution is to design an estimation function that guides \AStar through these constraints and finds the actual optimum efficiently.

\subsubsection*{Overview of the estimation function}\ \\
\label{costestimation} 

The general principle of \AStar is to explore the search space by generating a set of successors of the current node (in our case, a configuration) whose remaining costs towards destination are \emph{estimated} by a custom function. The most promising nodes are explored first (typically, based on a priority queue). Then, \AStar is guaranteed to find the optimal solution (i.e., path in the search space) as long as the provided function does not overestimate the cost of a node. Furthermore, as is usual for \AStar, the running time of the algorithm is primarily determined by the accuracy of the estimation.

Our estimation function relies on the following key idea.
Given a tour $\pi$ (in $d$-dimensions), a trajectory realizing this tour must have at least the cost of a trajectory realizing the projection of this tour separately in each dimension. Thus, one can use the maximum over these one-dimensional costs as an estimation for $\pi$. 
Consider for example the $2$-dimensional tour $\pi=((5, 10),$ $(10, 12),$ $(14, 7),$ $(8, 1),$ $(3, 5),$ $(5, 10))$ illustrated on Figure~\ref{fig:projection}. 
\begin{figure}[h]
	\centering
	\begin{tikzpicture}[scale=.55]
	\draw[very thin, lightgray] (-1.5,-1.3) grid (15.5,13.3);
	
	\draw[very thick, ->] (-1, 0) to (15,0);
	\draw[very thick, ->] (0, -1) to (0,13);
	
	% TOUR
	\tikzstyle{every node}=[circle, fill = white, draw=black, font=\scriptsize, inner sep= .3mm]
	\path (3, 5) node (p5) {\,};
	\path (8, 1) node (p4) {\,};
	\path (14, 7) node (p3) {\,};
	\path (10, 12) node (p2) {\,};
	\path (5, 9) node (p1) {$~$};
        \path (p1) node[draw=none,right=5pt] {\footnotesize{init}};
	\draw[->,thick] (p1)--(p2);
	\draw[->,thick] (p2)--(p3);
	\draw[->,thick] (p3)--(p4);
	\draw[->,thick] (p4)--(p5);
	\draw[->,thick] (p5)--(p1);
	
	% X-TOUR
	\tikzstyle{every node}=[circle, fill = white, draw=black, inner sep= .4mm]
	\path (5, 0) node (px1) {$~$};
	\path (3,0) node (px5) {\,};
	\path (14, 0) node (px3) {\,};
	\draw[dotted, red, thick] (p1) to (px1);
	\draw[dotted, red, thick] (p3) to (px3);
	\draw[dotted, red, thick] (p5) to (px5);
	\tikzstyle{every path}=[red, semithick, shorten >= 1pt,shorten <= 1pt, ->, bend right=30]
	\draw (px1) to (6,-0.2);
	\draw (6,-0.2) to (8,-0.2);
	\draw (8,-0.2) to (11,-0.2);
	\draw (11,-0.2) to (13,-0.2);
	\draw (13,-0.2) to (px3);
	\draw (px3) to (13,0.2);
	\draw (13,0.2) to (11,0.2);
	\draw (11,0.2) to (8,0.2);
	\draw (8,0.2) to (6,0.2);
	\draw (6,0.2) to (4,0.2);
	\draw (4,0.2) to (px5);
	\draw (px5) to (4,-0.2);
	\draw (4,-0.2) to (px1);
	\tikzstyle{every path}=[semithick, shorten >= 1pt,shorten <= 1pt]
	\draw[red, ->] (px3) to[out=-45, in=45, looseness=7] (px3);
	\draw[red, ->] (px5) to[out=135, in=-135, looseness=7] (px5);
	\path (px1) node[draw=none,fill=none,inner sep=.7mm] (px1loop){};
	\draw[red, ->] (px1loop) to[out=-110, in=-70, looseness=8] (px1loop);
	
	% Y-TOUR                
	\path (0, 1) node[circle, fill = white, draw=black, inner sep= .4mm] (py4) {\,};
	\path (0, 12) node[circle, fill = white, draw=black, inner sep= .4mm] (py2) {\,};
	\path (0, 9) node[circle, fill = white, draw=black, inner sep= .4mm] (py1) {$~$};
	\draw[dotted, blue, thick, shorten >= 5pt] (p1) to (py1);
	\draw[dotted, blue, thick] (p2) to (py2);
	\draw[dotted, blue, thick] (p4) to (py4);		
	\tikzstyle{every path}=[blue, semithick, shorten >= 1pt,shorten <= 1pt, ->, bend right=30]
	\draw (py1) to (0.2,10);
	\draw (0.2,10) to (0.2,11);
	\draw (0.2,11) to (py2);
	\draw (py2) to (-0.2,11);
	\draw (-0.2,11) to (-0.2,9);
	\draw (-0.2,9) to (-0.2,6);
	\draw (-0.2,6) to (-0.2,4);
	\draw (-0.2,4) to (-0.2,2);
	\draw (-0.2,2) to (-0.2,1);
	\draw (0.2,1) to (0.2,2);
	\draw (0.2,2) to (0.2,4);
	\draw (0.2,4) to (0.2,7);
	\draw[shorten >= 3pt] (0.2,7) to (0.2,9);
	
	\tikzstyle{every path}=[semithick, shorten >= 1pt,shorten <= 1pt]
	\draw[blue, ->] (py2) to[out=50, in=130, looseness=7] (py2);
	\draw[blue, ->] (py4) to[out=-130, in=-50, looseness=7] (py4);
	\path (py1) node[draw=none,fill=none,inner sep=.7mm] (py1loop){};
	\draw[blue, ->] (py1loop) to[out=-20, in=20, looseness=8] (py1loop);
	
	\end{tikzpicture}
	\caption{\label{fig:projection} Lower bounding the cost of a tour by its one-dimensional projections.}
      \end{figure}
      A vehicle realizing this tour must pass through the $x$-coordinates $(5,10,14,8,3,5)$ in this order. A similar statement holds for the $y$-coordinates. The sequence can be simplified further by retaining only the \emph{inflexion points}, namely the coordinates where the vehicle begins, ends, or switches direction at zero speed (in the considered dimension). Here, the inflexion points are $(5,14,3,5)$, their visit can be realized by the trajectory depicted in red along the $x$-axis. Again, a similar statement holds for the $y$-coordinates (in blue along the $y$-axis). Thus, the main work is to estimate the cost of such one-dimensional trajectories.
      
      The fact that the speed is zero at the inflexion points makes it possible to decompose the computation as a sum of independent costs~--~one for each consecutive pair (whose calculation is described in~\Cref{subsec:unidimensional}). This does not imply that the $d$-dimensional trajectory computed by \AStar will necessarily reach zero speed in the considered dimension, as other constraints may arise from combining the dimensions, which is fine because we are only lower bounding the costs here.

      The main difficulty is that, as \AStar progresses, what needs to be estimated is a \emph{suffix} of the initial tour rather than the entire tour. Also, the current speed of the vehicle may be different from zero. 
      Let $(x,dx)$ be the projection of the current position and velocity of the vehicle in the considered dimension, and let $x'$ be the position of the next inflexion point in the projection of the tour suffix. Either the vehicle is towards $x'$ (that is, $x'-x$ and $dx$ have the same sign), or it is moving in the opposite direction. In the first case, the vehicle must stop first, then come back to $x'$. Stopping takes $dx$ time units and $dx(dx+1)/2$ space units, then the calculation reduces to computing the cost between two points with starting and finishing speed zero, as above. In the second case, either the vehicle is able to stop without bypassing $x'$, or it is not. If it is not, we replace $x'$ in the projected suffix with the earliest position where the vehicle can stop, which is $x + dx(dx+1)/2$, at a cost of $dx$. The last case, where the vehicle is going towards $x'$ and is able to stop by $x'$ is dealt with in~\Cref{subsec:unidimensional}, showing that this case can again be reduced to a case with zero initial speed.

      In summary, the estimation of a suffix of the tour in a certain configuration reduces to computing, in each dimension separately (1) the cost from the current position and velocity to the next inflexion point (in the considered dimension), possibly changing the position of this point as explained above, and (2) the cost between all consecutive pairs of inflexion points in the projected suffix (in the considered dimension), each time starting and finishing at zero speed. 

\subsection{\EuclideanTSP tours are typically not optimal}
\label{sec:experiments}

In this section, we present experimental evidence that optimal tours for \EuclideanTSP are \emph{typically} not optimal for \VectorTSP. The experiments are conducted using the algorithmic components discussed in~\Cref{sec:high-level-algorithm} and~\ref{subsec:A-star} above. 
The instances are generated by placing cities uniformly at random within a two-dimensional square area. For each instance, we first use an external \EuclideanTSP solver (namely, \texttt{Concorde}~\cite{concorde}) to obtain an optimal \EuclideanTSP tour $\pi$, then we compute an optimal racetrack trajectory for this tour, $\texttt{racetrack($\pi$)}$, using the \multipointastar algorithm presented above. Then, in order to see if a better solution than $\texttt{racetrack($\pi$)}$ exists for this instance, we take $\pi$ as a starting point and try to improve this tour using the \flipTour heuristic (also presented above), with \multipointastar for computing the cost of the flipped tours.

Given the goal of the experiments, the only metric that we consider is how often the optimal tour for \EuclideanTSP turns out to be suboptimal for \VectorTSP. In other words, how often at least one flip leads to an improvement of the resulting trajectory. Two series of experiments were conducted in two-dimensional space. The results are averaged over $100$ iterations.
\begin{figure}[ht]
\centering
     \begin{subfigure}[t]{6cm}
	\includegraphics[scale=.45]{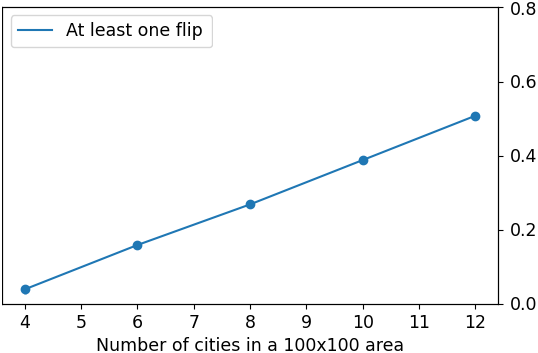}
         \label{fig:y equals x}
       \end{subfigure}
       \hspace{.5cm}
     \begin{subfigure}[t]{6cm}
	\includegraphics[scale=.45]{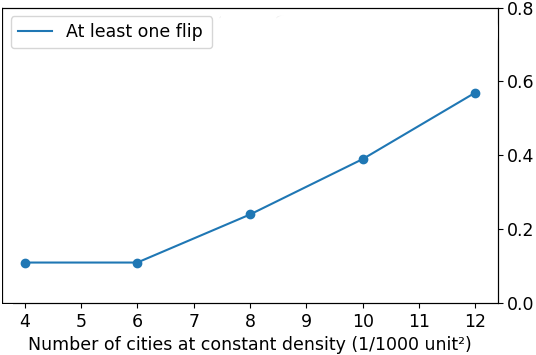}
     \end{subfigure}        
     \caption{\label{fig:plots}Probability that an \EuclideanTSP solution is improved (by a flip).}
\end{figure}
In the first experiment (\Cref{fig:plots}, left), the space is limited to a $100\times 100$ subset of $\mathbb{Z}^2$, in which $n$ cities are distributed uniformly at random, for $n$ ranging over $4,6,8,10,12$. In order to eliminate the potential effects of a restricted space, the second experiments (\Cref{fig:plots}, right) adapts the same scenario on a growing space, so that the density of cities is constant. The results suggest clearly that \EuclideanTSP tours tend to be suboptimal as the number of cities grows, as well as when the space is less restricted, with nearly $60\%$ suboptimal tours for instances of $12$ cities. 
Recall that the \flipTour algorithm is only a heuristic that stops when a local optimum is reached, whereas the reference \EuclideanTSP tour is a global optimum. As a result, our experiments can \emph{under-estimate} the discrepancy between \VectorTSP and \EuclideanTSP (which only makes our point stronger).

To conclude this section, \Cref{fig:improved-tour} shows a larger instance of $20$ cities and two corresponding trajectories; on the left side, an optimal trajectory realizing an optimal \EuclideanTSP tour~$\pi$ returned by the Concorde solver (with a cost of $128$); on the right side, an optimal trajectory realizing a different tour $\pi'$ obtained after two flips of $\pi$ (with a cost of $120$). 

\begin{figure}[t]
  \centering
  \begin{subfigure}[b]{0.44\textwidth}
    \includegraphics[scale=.45]{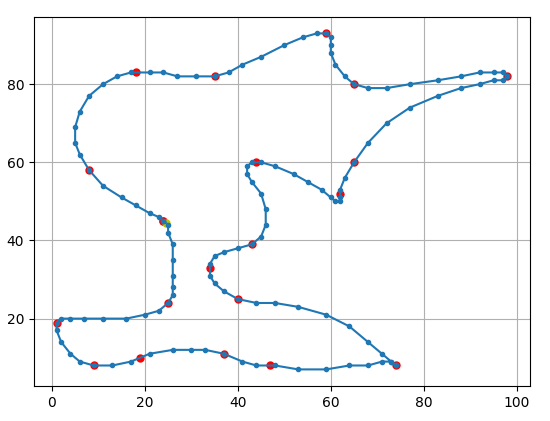}
  \end{subfigure}
  \begin{subfigure}[b]{0.44\textwidth}
    \includegraphics[scale=.46]{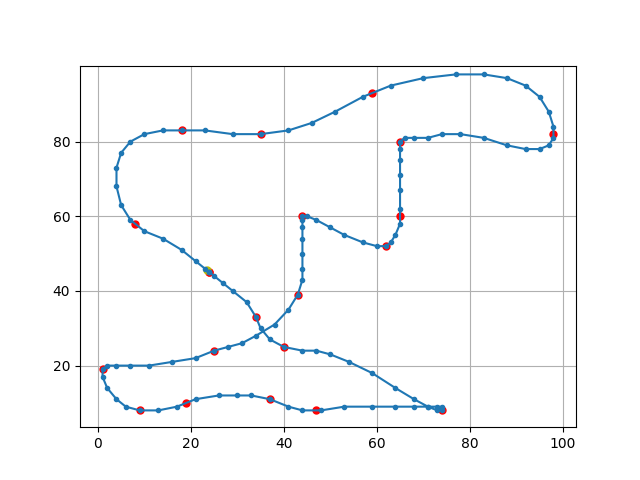}
  \vspace{-17pt}
  \end{subfigure}
  \caption{\label{fig:improved-tour}Example of an instance where the optimal \EuclideanTSP tour can be improved.}
\end{figure}

\section{Conclusion and future work}
\label{sec:conclusion}

In this paper, we have presented and studied a new version of the TSP called \VectorTSP, where the movements of the vehicle must obey initial constraints according to the Racetrack model. After proving a number of basic properties about this problem and about the Racetrack model, we showed that \VectorTSP is NP-hard. While not really surprising, the proof of this result turned out to be highly non-trivial. More direct transformation from existing problems may exist, we leave this as an open question.

We used several of the preliminary observations for devising and motivating an algorithmic framework based on two interacting layers, namely a high level tour computation algorithm, and an exact guided pathfinding algorithm for trajectory evaluation. As simple as it is, we showed through experiments that this framework outperforms optimal solutions based on the \EuclideanTSP, thereby motivating further study of the problem.

One question of interest is how far an optimal trajectory of an optimal \EuclideanTSP tour could be from the optimal \VectorTSP solution.
Our work also leaves a number of algorithmic questions open. For example, the high-level part of the algorithm relies on a standard TSP heuristics. Would there be dedicated heuristics for this component in \VectorTSP? On the other hand, our \multipointastar algorithm for trajectory computation could be of independent interest and it would be interesting to find other problems beyond TSP where such a component is useful. Indeed, the Racetrack model is appealing for its simplicity and amenability to algorithmic investigation. Several geometric, robotic, and mobility-related problems could certainly be revisited from a Racetrack perspective, generating new insights and potential for useful applications.

\bibliographystyle{splncs04}
\bibliography{paper} 

\end{document}